\documentclass[a4]{article}
\usepackage{authblk}
\usepackage{amsfonts} 
\usepackage{graphicx}
\usepackage{mathtools}
\usepackage{stmaryrd}
\usepackage{amsthm}
\usepackage{amsmath}
\usepackage{amssymb}
\usepackage{cite}

\DeclareMathOperator*{\argmin}{arg\,min}

\newtheorem{theorem}{Theorem}[section]
\newtheorem{lemma}[theorem]{Lemma}
\newtheorem{definition}[theorem]{Definition}

\title{Parameterized Neural Networks for Finance}
\author[1,*]{Daniel Oeltz}
\author[1]{Jan Hamaekers}
\author[2]{Kay F. Pilz}
\affil[1]{Fraunhofer Institute for Algorithms and Scientific Computing SCAI, Schloss Birlinghoven, 53757 Sankt Augustin, Germany.}
\affil[2]{kinetic mind GmbH}
\affil[*]{Corresponding author: Daniel Oeltz, daniel.oeltz@scai.fraunhofer.de}
\begin{document}
\maketitle
\begin{abstract}
We discuss and analyze a neural network architecture, that enables learning a model class for a set of different data samples rather than just learning a single model for a specific data sample. In this sense, it may help to reduce the overfitting problem, since, after learning the model class over a larger data sample consisting of such different data sets, just a few parameters need to be adjusted for modeling a new, specific problem. After analyzing the method theoretically and by regression examples for different one-dimensional problems, we finally apply the approach to one of the standard problems asset managers and banks are facing: the calibration of spread curves. The presented results clearly show the potential that lies within this method. Furthermore, this application is of particular interest to financial practitioners, since nearly all asset managers and banks which are having solutions in place may need to adapt or even change their current methodologies when ESG ratings additionally affect the bond spreads.

\end{abstract}
\section{Introduction}
Deep learning has shown impressive results over the past two decades in various fields, such as image recognition and classification, as well as natural language processing. A lot of training data is usually needed to calibrate neural networks for these tasks, but, unfortunately, financial data is often quite limited. For example, consider calibrating a neural network for forecasting the distribution of returns corresponding to a certain stock conditioned on past returns. Although there may be a long history of daily closing stock prices, maybe even 40 years for certain companies, we end up with only approximately 10.000 data points, which is not that much for calibrating a neural network, in particular, because such problems often have a low signal-to-noise ratio. Moreover, the time series may not be stationary either, which means that past data may not be a good input for training a network that predicts future data. Hence, the total amount of data that is suitable for fitting, has to be further reduced by selecting sub-periods of the available data. On the other hand, there often exists a global structure in the data, that can be found in all periods, maybe even over different stocks, and the question is, how can we make use of such structures when training the network. 

Here, multi-task learning (MTL) comes into play. Multi-task learning has been successfully applied in different fields such as time series forecasting, in general \cite{Nikentari2022}, weather forecasting and power generation modeling \cite{Lencione2021, DoradoMoreno2020, Schreiber2021, Schreiber2022}, computer vision \cite{Girshick2015, Liu_2019_CVPR}, natural language processing \cite{Collobert2008, Lu_2020_CVPR} and many other applications, see \cite{Caruana1996AlgorithmsAA}. MTL can improve the generalization capabilities resulting in a lower risk of overfitting, see \cite{Vafaeikia2020, Vandenhende2020, Baxter2019, Baxter1997}. Additionally, learning new tasks may be faster and more robust using MTL. A drawback of these methods is the so-called \emph{negative transfer} which describes the effect of getting larger errors in less difficult tasks, compared to single neural network models, due to the larger errors involved by the difficult tasks \cite{Nikentari2022, DoradoMoreno2020, Lee2017, Liu2019a}.

In this paper we discuss and analyze the application of a very simple MTL architecture as introduced in \cite{Schreiber2021} under the name task embedding network, which we believe is a promising approach to solving several problems in finance. Note, that we prefer the name parameterized neural network (PNN) in our context, for reasons that are discussed later.

The paper is organized as follows. In section \ref{MTL-Overview} we give a brief overview of different MTL architectures and describe the PNN design. We also discuss the generalization capability of the PNN following the approach of \cite{Baxter2019}. Section \ref{Numerical-Experiments} presents several simple experiments that give insight into the functioning of the PNN and its performance on certain problem classes. The section concludes with a more complex example inspired by the problem of calibrating spread curves to bond market quotes. In the final section, the results are summarized, and an outlook on potential applications in the area of Finance as well as future research are given.

\section{Multi-Task Learning}\label{MTL-Overview}
\subsection{Architectures}
In general, MTL can be subdivided into two categories \cite{Vandenhende2020},
\emph{soft parameter sharing} and \emph{hard parameter sharing}. In soft parameter sharing, we calibrate a network to each task separately in a way, that the network parameters are somehow related, e.g. by penalizing deviations between parameters of different networks \cite{Duong2015, Yang2016}.
Hard parameter sharing goes back to \cite{Caruana1993} and PNNs belong to this category. In hard parameter sharing, the neural networks for every single task share a certain subset of their parameters. Here, a lot of different architectures exist \cite{Vafaeikia2020, Ruder2017}, and as discussed in \cite{Vandenhende2020} may be further categorized into \emph{encoder-based} and \emph{decoder-based} architectures. Encoder-based architectures share the input and first layers (bottom layers) of the networks.  This encoder-based approach is also known as internal representation learning and \cite{Baxter2019} was able to show that this lowers the risk of overfitting compared to calibrating a single network to each task separately.
Decoder-based architectures, which PNNs belong to, apply task-specific networks using their output as input for a single network that is task-independent. 

The selection of the architecture may depend on different considerations, and given a problem (to our knowledge) there are no general criteria available that determine which architecture is more suitable. Encoder-based approaches seem to dominate the field of multi-task learning, especially in the area of computer vision. Recent work compares both architectures on different kinds of problems \cite{Vandenhende2020, Zhang2021}.
However, as we will discuss in the next section, we think that many problems in finance can benefit from the very simple multi-task architecture that is discussed in the following.

\subsection{Parameterized Neural Networks}\label{sec-task-embedding}

\begin{figure}
\centering
\includegraphics[scale=0.5]{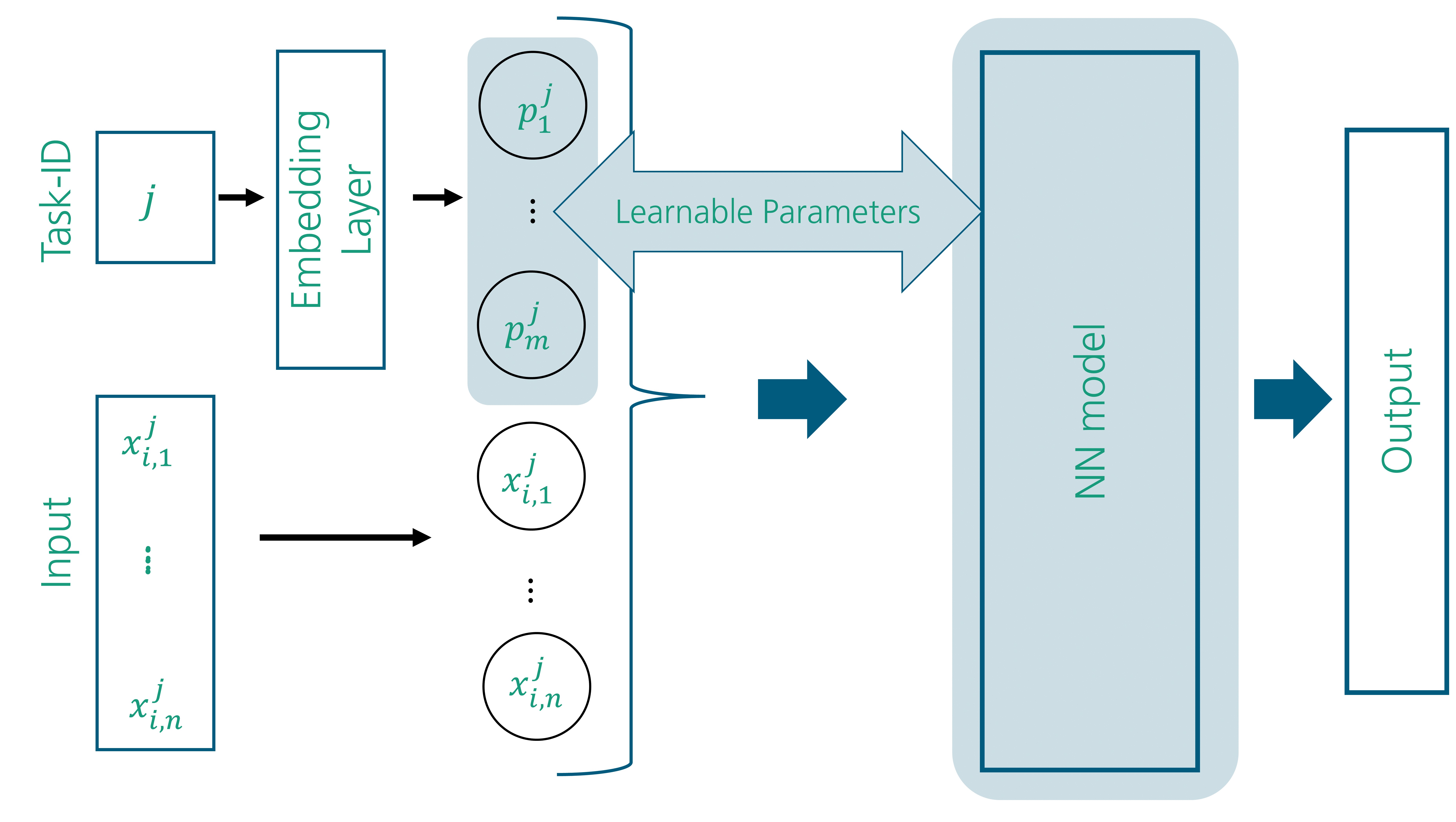}
\caption{Parameterized network architecture. The input is concatenated with a task-specific vector that is individually trained while the remaining network architecture is unchanged (except for a higher dimension in input space).} \label{fig-tast_embedding-architecture}
\end{figure}

We discuss the PNN architecture and motivate why it is a good candidate for MTL in the financial field. To our knowledge, PNNs were proposed in \cite{Schreiber2021} the first time, under the name task embedded networks. The tasks are represented by integers, very similar to word embedding \cite{Mikolov2013}, and an embedding layer is used to map these integers to a corresponding parameter vector consisting of real values. The concatenation of this parameter vector with the original input data defines the input to the main neural network, which is the same for all tasks (in terms of the network weights), see figure \ref{fig-tast_embedding-architecture}. Each of these parameter vectors is optimized during training by optimizing the embedding layer weights.

To illustrate this approach, let us assume we have $n$ tasks, and a task $1\leq i\leq n$ consists of $m$ samples $\{z_{ij}\}_{j=1}^{m}$, where $z_{ij} = (x_{ij},y_{ij}) \in Z\subset \mathbb{R}^{k_1}\times \mathbb{R}^{k_2}$ for task independent $k_1$ and $k_2$. For ease of notation, each task has the same number of samples. The $x_{ij}$ are the (regression or approximation) function inputs, and the $y_{ij}$ are the outputs.

Further, assume the embedding layer has dimension $l$, which means that we have a task-specific parameter vector $p_i\in\mathbb{R}^l$ for each task $1\leq i\leq n$. Given a task $i$, the neural network is a function $g:\mathbb{R}^{k_1} \rightarrow \mathbb{R}^{k_2}$, that also can be regarded as $g(x;p_i)$, the approximating function of task $i$. More generally, for any parameter vector $p$ (not necessarily corresponding to a task from the training set) the neural network $g(x;p)$ represents a family of functions depending on the parameter $p$. 

The reason why we prefer the term parameterized neural network in contrast to the term task-embedded neural network is, that the problems we have in mind are not related to the solution of different tasks (as in classical multi-task learning), but rather related to finding an optimal solution within a family of parameterized functions.

It is quite obvious that this approach is very generic and independent of the basic network architecture that uses task embedding as input.
Therefore, without many changes, one can directly apply this approach not only to simple feed-forward multi-layer networks but also to more sophisticated architectures such as generative methods (VAE and GAN), mixture-density networks (MDN), and recurrent networks, as well as to reinforcement learning. We will give an outlook of promising applications in the area of finance at the end of this work.
Although being very generic on the one hand, the approach is restrictive in terms of the structure of the different tasks. Obviously, the architecture makes only sense for multiple tasks with the same (or similar) input and output spaces.  
This might be a substantial limitation for some use cases of multi-task learning (even in finance), but we believe that there are a lot of applications to which this approach can be successfully applied, and does significantly help to get stable models based on a limited amount of available training data.

We discuss certain properties of this approach in the following sub-sections.

\subsubsection{Fast Calibration to New Data}
After model calibration, we end up with a family of functions parameterized by the parameter vector $p$. Whenever we have to calibrate the model to new data represented by a new task integer, we just have to find a new parameter $p\in \mathbb{R}^l$ and get a model for this task. 
The start value for the parameter is usually important for the speed of convergence when using stochastic gradient methods. Here, one may use either the parameter of the last calibration (for instance, when the tasks are generated by a time series and expected to be auto-correlated) or simply the average over all parameters that have been calibrated so far. In our numerical experiments, we use the latter approach to calibrate to the test data, where the average over parameters is taken from the training data. This approach gave good results just after a few steps of gradient descent for our experiments.

\subsubsection{Interpretability and Validation}
Since we have a parameterized family of functions, we can write the neural network model as $g(x,p)$, the function of the inputs $x \in \mathbb R^{k_1}$ and the parameter $p \in \mathbb R ^l$. If $g$ is continuous on $\mathbb{R}^{k_1 + l}$, it is Lipschitz continuous on each compact subset $\mathcal{C}$ and therefore it easily follows that for any two compact subsets $\mathcal{C}_1\subset \mathbb{R}^{k_1}$ and  $\mathcal{C}_2\subset \mathbb{R}^{l}$ there is a constant $L(\mathcal{C}_1,\mathcal{C}_2)$ such that 
\[
\|g(x;p_1)-g(x;p_2)\|\leq L(\mathcal{C}_1,\mathcal{C}_2)\|p_1-p_2\| \mbox{ for all } x \in \mathcal{C}_1, p_1,p_2\in \mathcal{C}_2.
\]
Hence, given two parameters $p_1$ and $p_2$ we directly get an upper bound on the distance between the two models corresponding to $p_1$ and $p_2$. We could use techniques to compute bounds on the Lipschitz constant, see \cite{Scaman2018} and the references therein, for a given network $g$ or use methods to build Lipschitz-constrained networks \cite{Gouk2020} to explicitly bound the Lipschitz constant.  But even if we do not explicitly know the Lipschitz constant, this property may help to understand and validate new model parameters derived from training on new data by comparing them to previous results and considering tasks that were similar in the past. As we will see in the numerical experiments the parameters may also be used to identify certain regimes in the tasks and allow to cluster them.

It is another advantage of the PNN that a set of extensively validated and tested models with parameters in a certain range, $p \in [p_{\rm low}, p_{\rm up}]$, transfer their validity to new models with calibrated parameters in the same range. This may save computational costs and time for fully re-testing these models.

\subsubsection{Separation of Regularization}
The PNN allows separating the regularization regarding the tasks and the parameters. For instance, we may apply Gaussian noise to the $x$ inputs but leave the parametrization $p$ unchanged. If we penalize the first derivatives with respect to the inputs $x$ by incorporating them into the network output and cost function, as in \cite{Huge2020}, we can enforce different restrictions to the $x$ features than for the parameters. This gives additional control and flexibility which may be useful in certain situations.

\subsection{Theoretical Considerations} \label{Theory}

\begin{figure}\centering 
\includegraphics[width=0.5\textwidth]{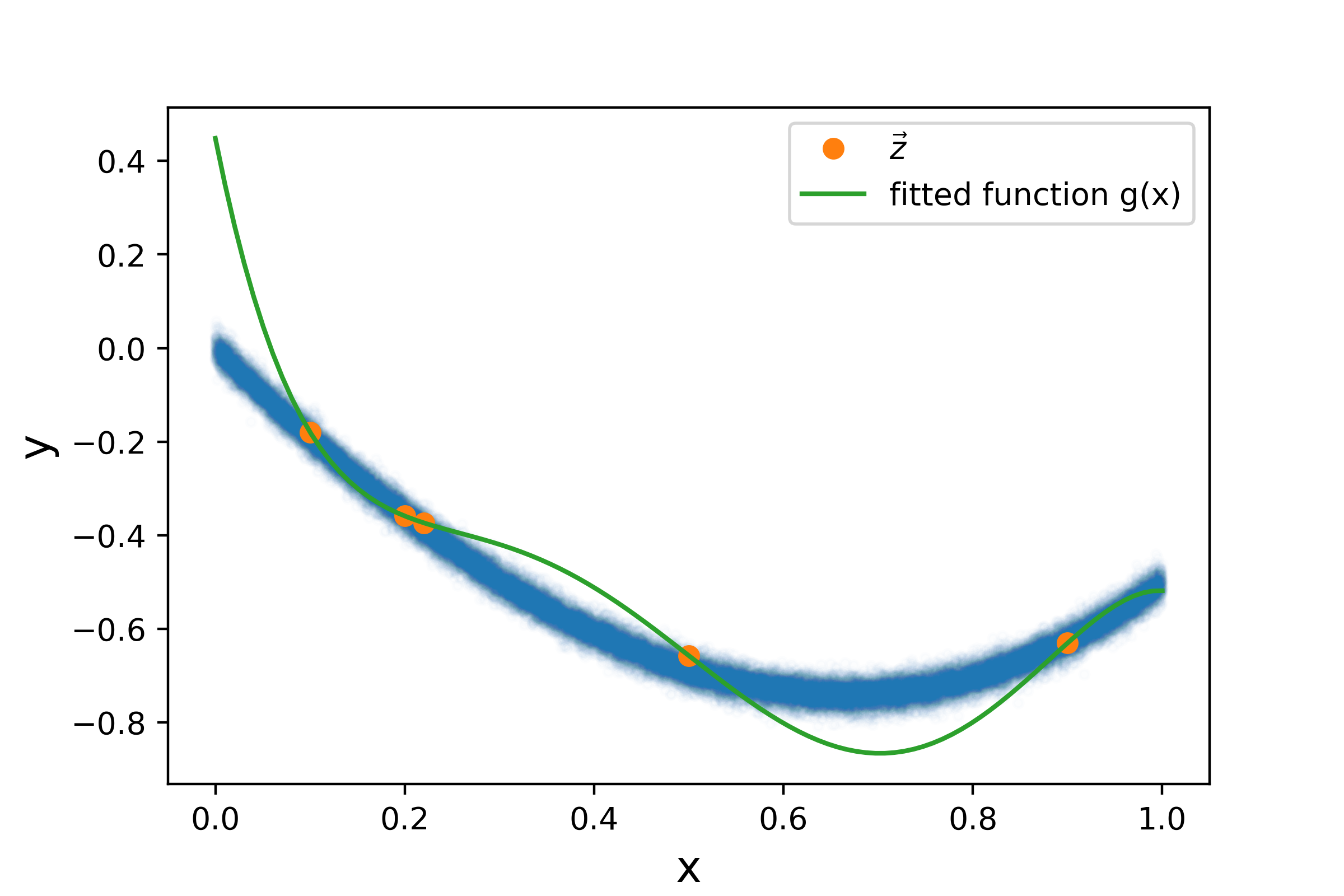}
\caption{Simple example for a case with overfitting, where the empirical loss ($\langle l_g\rangle_{\vec{z}}=0$) and the true loss ($\langle l_g\rangle_P\approx 0.01$) differ substantially. The blue dots represent samples from $P$ while the orange circles mark the training points. The loss function $l(y_1,y_2):=|y_1-y_2|^2$ is the mean squared error.} \label{Fig-True-vs-Empirical-Loss}
\end{figure}
In this section, we discuss the generalization property of the PNN and introduce some notations. The main result of this section is theorem \ref{Generalization-Property} that gives insight into the impact of using MTL compared to calibrating a model to a single task. Here, we use mainly the results from \cite{Baxter2019a} on representation learning and adapt the approach to the PNN.

To introduce the notation and the basic principle, let us first consider the case of learning one task, i.e. data consisting of vectors $\vec{z}_j = (x_j,y_j) \in \mathbb{R}^{k_1}\times\mathbb{R}^{k_2}$. We assume that the training set $\vec{z}$ is created by a probability distribution $P$ on $Z\subset \mathbb{R}^{k_1}\times\mathbb{R}^{k_2}$ and define $Z^m$ as the set of all samples of length $m$ according to $P$, such that $\vec{z} \in Z^m$.
Let $l:\mathbb{R}^{k_2}\times \mathbb{R}^{k_2}\rightarrow \mathbb [0,M]$ be a \emph{loss function} with fixed $M>0$. For a function $g: \mathbb{R}^{k_1} \rightarrow \mathbb{R}^{k_2}$ we define the empirical loss by
\begin{equation}\label{Empirical-loss}
    \langle l_g\rangle_{\vec{z}}:=\frac{1}{m}\sum_{j=1}^m l(y_j,g(x_j))
\end{equation}
Consider $g(x;\theta)$ a neural network with parameters $\theta\in \Theta$. The typical learning task is to determine $\theta^\star$ s.t.
\begin{equation}
    \theta^\star = \argmin_{\theta\in \Theta} \langle l_g \rangle_{\vec{z}}
\end{equation}
which defines a learning algorithm 
\begin{equation}\label{Learning-Algo}
    \mathcal{A}:\bigcup_{m\geq 1} Z^m \rightarrow \{g(x,\theta)\mid \theta\in \Theta\}=:\mathcal{H}.
\end{equation}
By $l_{\mathcal H}$ we denote the family of loss functions that are defined by all $g(\cdot ) \in \mathcal H$.
In equation (\ref{Empirical-loss}) the empirical loss is computed on the training set only and may overestimate the model performance due to overfitting, which means that although the empirical loss is small, the true loss defined by
\begin{equation}
    \langle l_g\rangle_P:=\int_Z l(y,g(x))dP
\end{equation}
can be quite large. As an example see figure \ref{Fig-True-vs-Empirical-Loss}, which shows a fit with zero empirical loss but a large true loss, hence is overfitting. For having a model that generalizes to new data, it is quite essential that the empirical loss used in the learning algorithm defined by (\ref{Learning-Algo}) is close to the true loss. Statistical learning theory provides bounds for the difference between these two losses, depending on the complexity of the learning model and the number of training points. To measure the distance between both losses, we use a family of metrics $d_\nu:\mathbb{R}_+\times\mathbb{R}_+$ introduced in \cite{Haussler1992},
\[
d_\nu (x,y):=\frac{|x-y|}{\nu+x+y}.
\]
One gets the following upper bound under suitable conditions for $\nu>0$ and $0<\alpha<1$,
\begin{equation}
    {\rm Pr}\{\vec{z}\in Z^m: \exists l_g \in l_{\mathcal{H}}: d_\nu\left( \langle l\rangle_P, \langle l\rangle_{\vec{z}}\right)>\alpha\} \leq C(\alpha,\nu, \mathcal{H}) e^{-\frac{\alpha^2\nu m}{8M}},
\end{equation}
where $C(\alpha,\nu, \mathcal{H})$ is a constant depending on $\alpha$, $\nu$ and the so-called $\varepsilon$-capacity of $\mathcal{H}$, the set of neural networks defined in (\ref{Learning-Algo}), see \cite{Baxter2019a} for further details. 
Therefore, to guarantee with probability $\delta$ that the difference between empirical and true loss does not differ more than $\alpha$ with respect to $d_\nu$, it suffices to have $m$ training points, with
\[
m>\frac{8M}{\alpha^2\nu}\ln \left( \frac{C(\alpha,\nu, \mathcal{H})}{1-\delta}\right).
\]
For a proof of this bound and further details, see \cite{Baxter2019a} and the references therein. 
We will now consider the case of a PNN with fixed parameter dimension $l$ and multiple tasks. Recall that the PNN calibrated to $n$ tasks, where each task has $m$ data points, produces a sequence of $n$ different functions $\vec{g}:=\left(g(x;p_i,\theta)\right)_{i=1,\ldots ,n}$ that share the same network parameters $\theta$ and do only differ in their input parameter (concatenated to the original input) $p_i\in\mathbb{R}^l$. For ease of notation, we simply write $g_{p_i}$ instead of $g(x;p_i,\theta)$. If we denote the training points by $\vec{z}$ where $\vec{z}_i$ denotes the training data of the $i$-th task sampled from a distribution $P_i$, we define the empirical loss analogously to the previously discussed case with only one task,
\begin{equation}\label{empirical-loss_mt}
    \langle l_{\vec{g}}\rangle_{\vec{z}}:=\frac{1}{n}\sum_{i=1}^n  \langle l_{g_{p_i}}\rangle_{\vec{z}_i},
\end{equation}
and the true loss
\begin{equation}\label{true-loss_mt}
    \langle l_{\vec{g}}\rangle_{\vec{P}}:=\frac{1}{n}\sum_{i=1}^n  \langle l_{g_{p_i}}\rangle_{P_i}.
\end{equation}
In order to apply a similar approach as in \cite{Baxter2019a}, we define $\mathcal{F}:=\left\{ f(x)=(x,p)\mid p \in \mathbb{R}^l\right\}$, $\mathcal{G}:=\left\{g(x,p;\theta)\mid \theta \right\}$ and furthermore 
\begin{equation}
\mathcal{F}^n :=\left\{f(x_1,\hdots,x_n):=(f_1(x_1),\hdots , f_n(x_n))\mid f_i \in \mathcal{F}, x_i\in\mathbb{R}^{k_1} \right\}
\end{equation}
and
\begin{equation}
    \bar{\mathcal{G}}:=\left\{g((x_1, p_1,\hdots , x_n,p_n);\theta):=\left(g(x_1,p_1;\theta),\hdots , g(x_n,p_n;\theta)\right)\mid \theta \right\}.
\end{equation}

Using this notation we obtain the following theorem.

\begin{theorem}\label{Generalization-Property} Let $\nu>0$, $0<\alpha<1$, be fixed and $\varepsilon_1, \varepsilon_2 > 0$ such that $\varepsilon_1+\varepsilon_2=\frac{\alpha\nu}{8}$. 
For $0<\delta<1$ and the structure 
\[ 
X^n\xmapsto{\mathcal{F}^n}V^n\xmapsto{\bar{\mathcal{G}}}\mathbb{R}^{k_2\cdot n}
\] and $\vec{z} \in Z^{(m,n)}$ be generated by 
$m>\frac{8M}{\alpha^2\nu}\left[ \ln(\mathcal{C}(\varepsilon_1,\mathcal{F})) + \frac{1}{n}\ln \frac{4\mathcal{C}(\varepsilon_2,l_{\mathcal{G}})}{\delta} \right]$ independent samples
we have
\begin{equation}
    {\rm Pr} \left\{ z \in Z^{(m,n)}:\exists \bar{g}\circ\vec{f}\in \bar{\mathcal{G}}\circ\mathcal{F}^n: d_{\nu}(\langle l_{\bar{g}\circ \vec{f}}\rangle_{\vec{z}}, \langle l_{\bar{g}\circ \vec{f}}\rangle_{\vec{P}}) > \alpha\right\}\leq \delta
\end{equation}
\end{theorem}
A sketch for a proof is given in Appendix \ref{Appendix-proof}.

This theorem shows that the parameterized network approach may reduce overfitting compared to the single-task case. We see that the number of tasks reduces the term involving the complexity of the overall set of network functions.
\section{Numerical Experiments}\label{Numerical-Experiments}
In this section, we investigate the behavior and performance of the proposed method using simulation-based experiments. If not stated otherwise, we use the following settings:
\begin{itemize}
\item The base network of the PNN has 3 inner layers with 32 neurons and SELU activation functions.
\item Adam optimizer with 8000 epochs and exponentially decaying learning rate.
\item 100 tasks used to train the network.
\item 250 tasks to test the method.
\end{itemize}
We measure the error by the standard mean squared error for a single task $i$ on the test data,
\[
e_i :=\frac{\sqrt{ \sum_{j}\|y_{ij}-\hat g(x_{ij}, \hat p_i)\|^2}}{m},
\]
where $g$ is the resulting parameterized network fitted on the training data and $\hat p_i$ denotes the parameter that is calibrated to the test data. 
The calibration of the parameter $\hat p_i$ is different from the common approach of measuring error on a test data set without recalibrating anything on the data. We also measured the error w.r.t the data generation process of each task in the training data to create test data without recalibrating the respective parameter for this task to the new data. We use this methodology since one of our main interests in this approach is the capability of the model to calibrate to new unseen data. For the calibration of $\hat p_i$ we use the mean over all  parameters from the training as the start value on the training data and Adam optimizer with 100 epochs, a batch size of 10, and a learning rate of 0.01.
We generate 250 tasks for the training data and define the training error as the square root of the mean overall errors of the single tasks
\[
e := \frac{\sqrt{\sum_{i=1}^{250} e_i^2}}{250}.
\]
For all one-dimensional experiments below we use a uniform grid with 100 gridpoints on the function domain to construct the test data.

\subsection{Family of quadratic functions}

\begin{figure}
\includegraphics[width=0.5\textwidth]{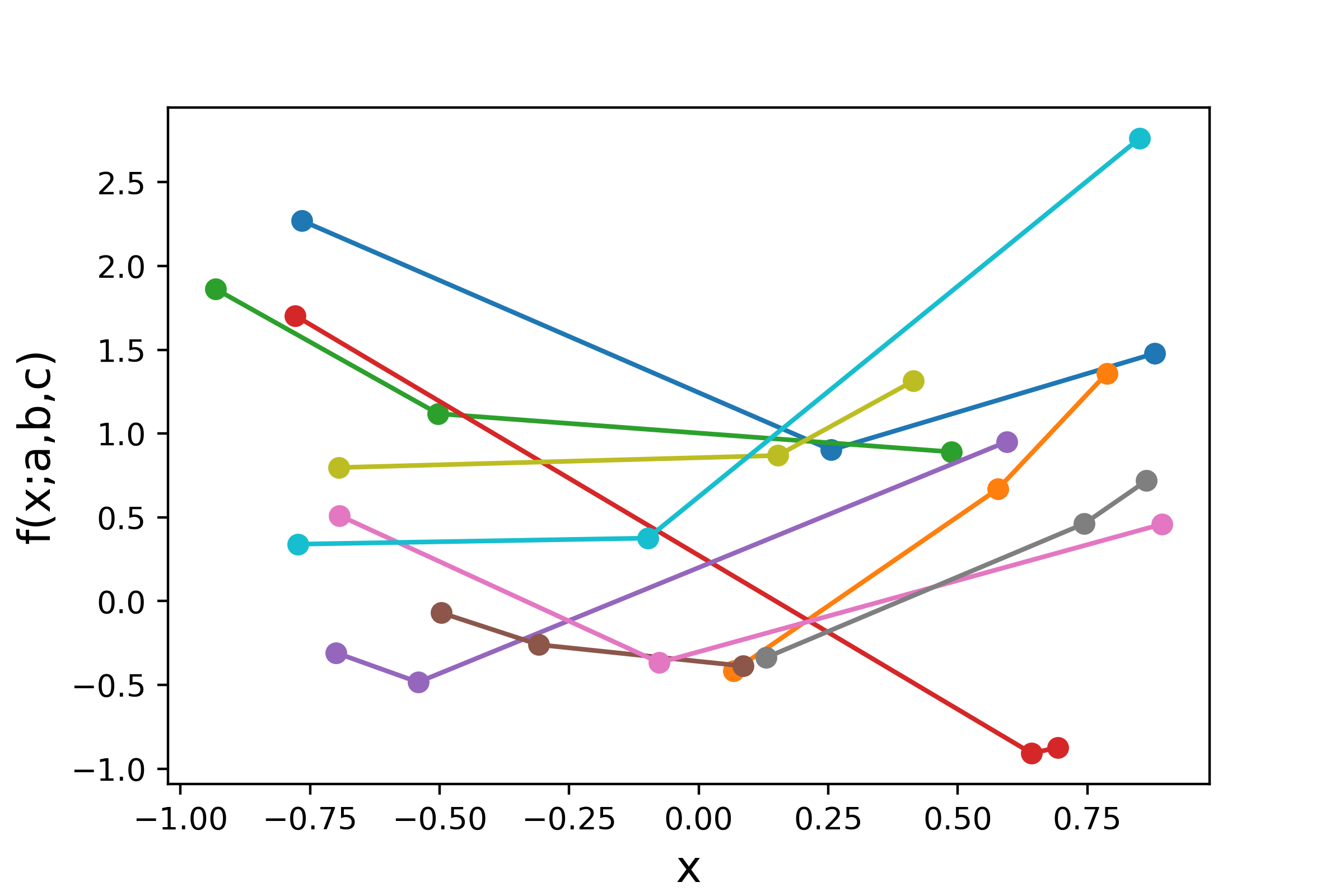}
\includegraphics[width=0.5\textwidth]{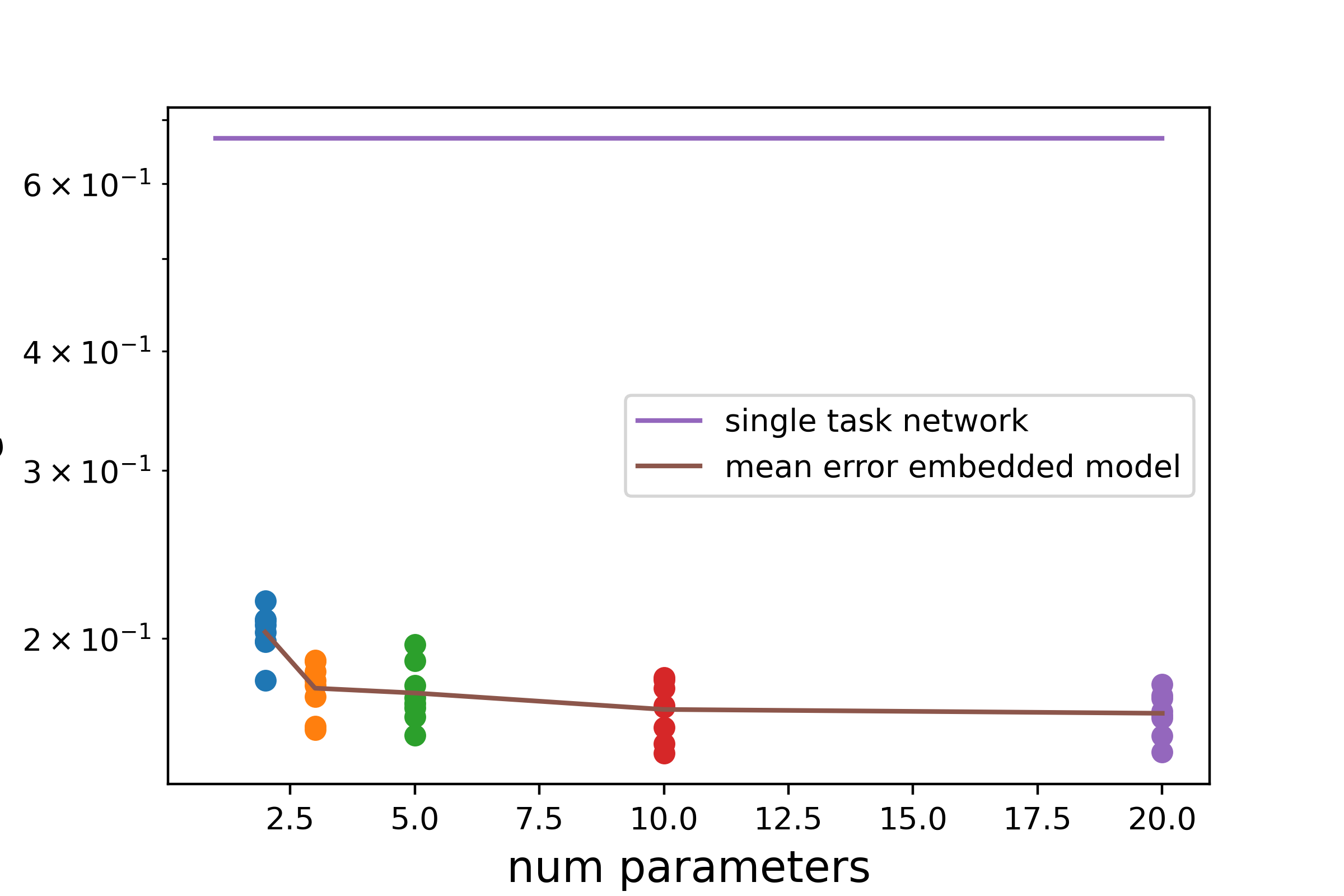}
\caption{Left: PNN tasks sampled from the family of quadratic function defined in (\ref{eq-quadratic}). Right: Approximation errors depending on the number of parameters used in the PNN, for 100 different tasks used in training.} \label{fig-quadratic}
\end{figure}

\begin{figure}
\includegraphics[width=0.5\textwidth]{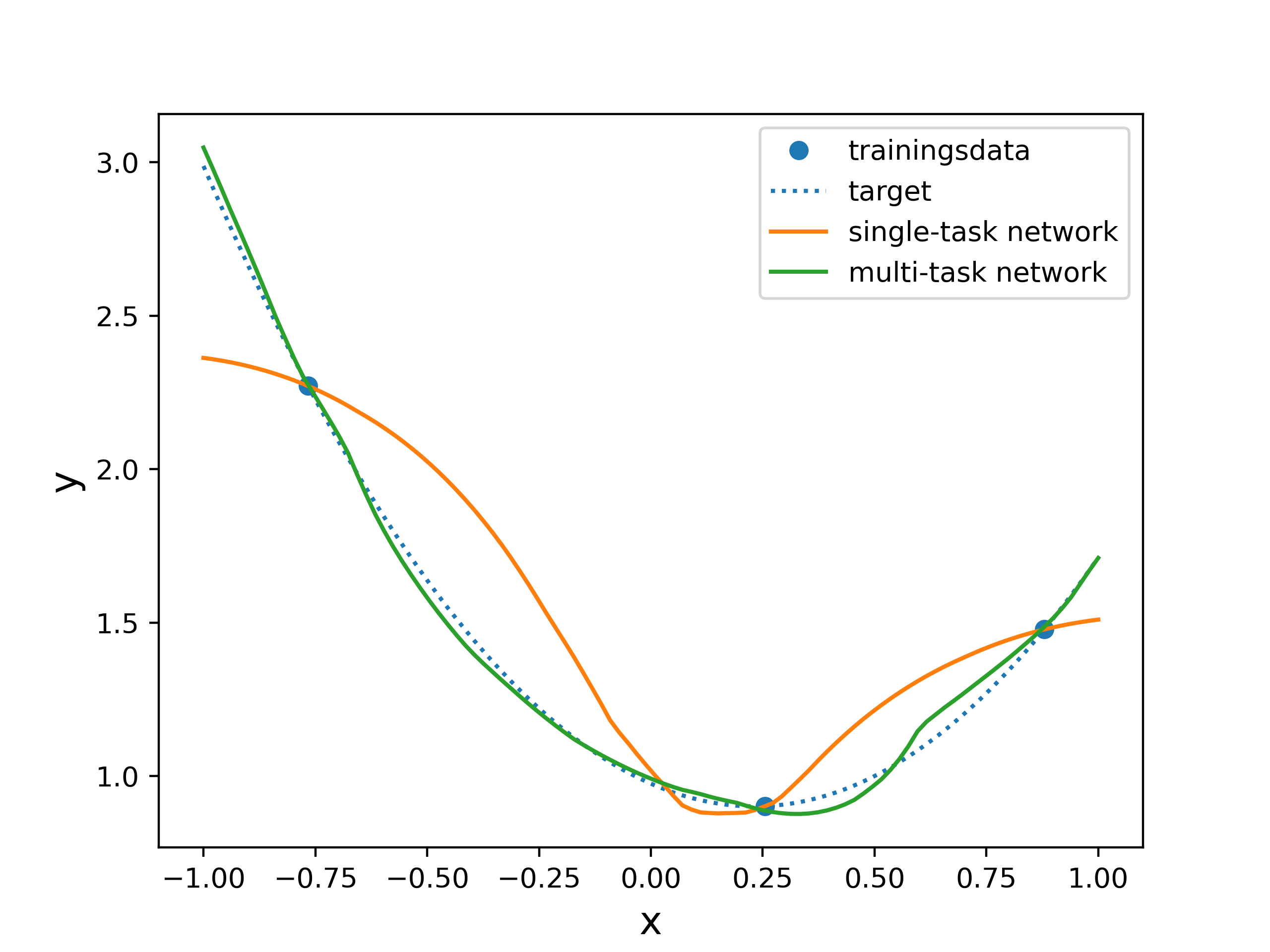}
\includegraphics[width=0.5\textwidth]{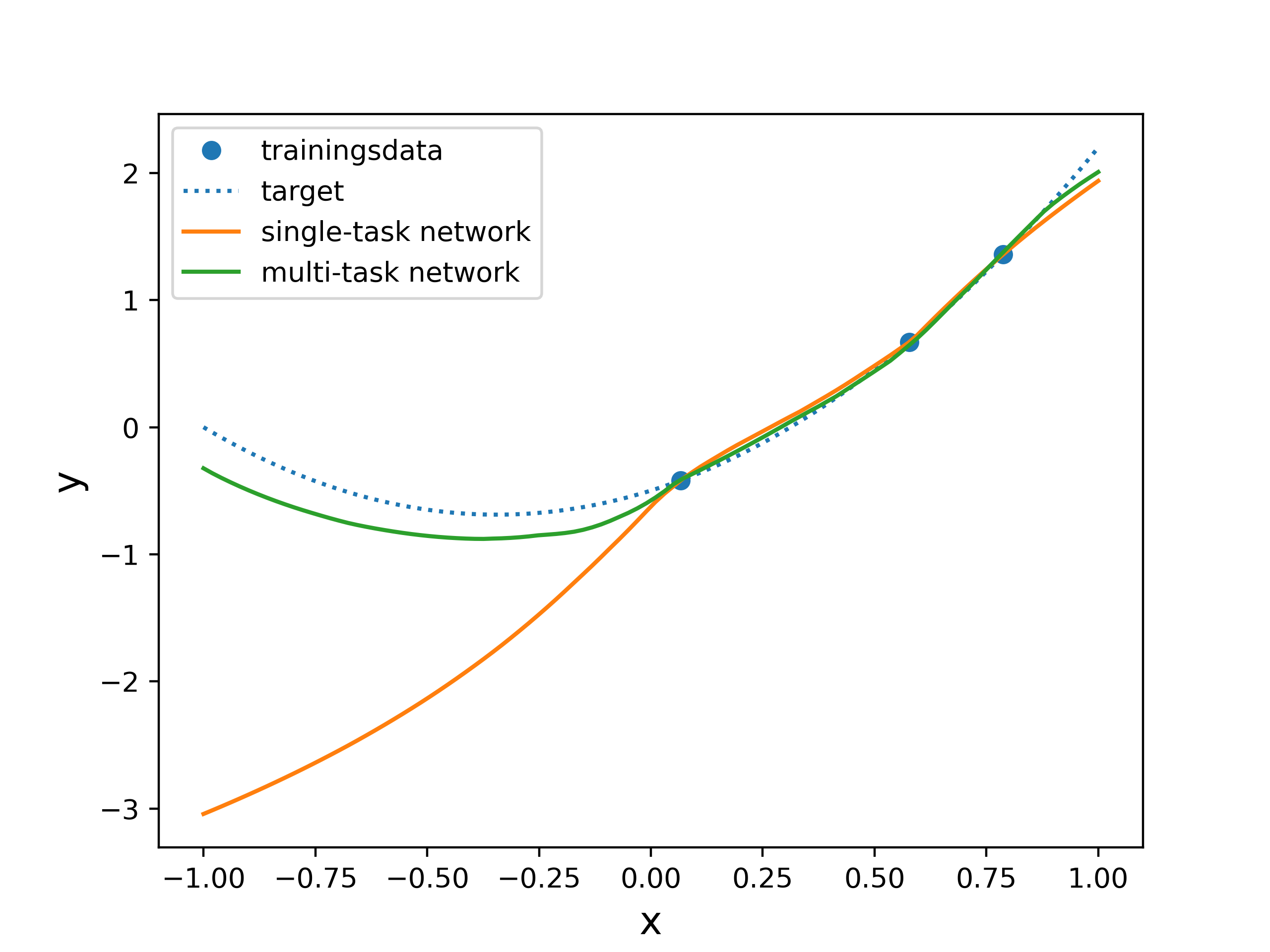}
\includegraphics[width=0.5\textwidth]{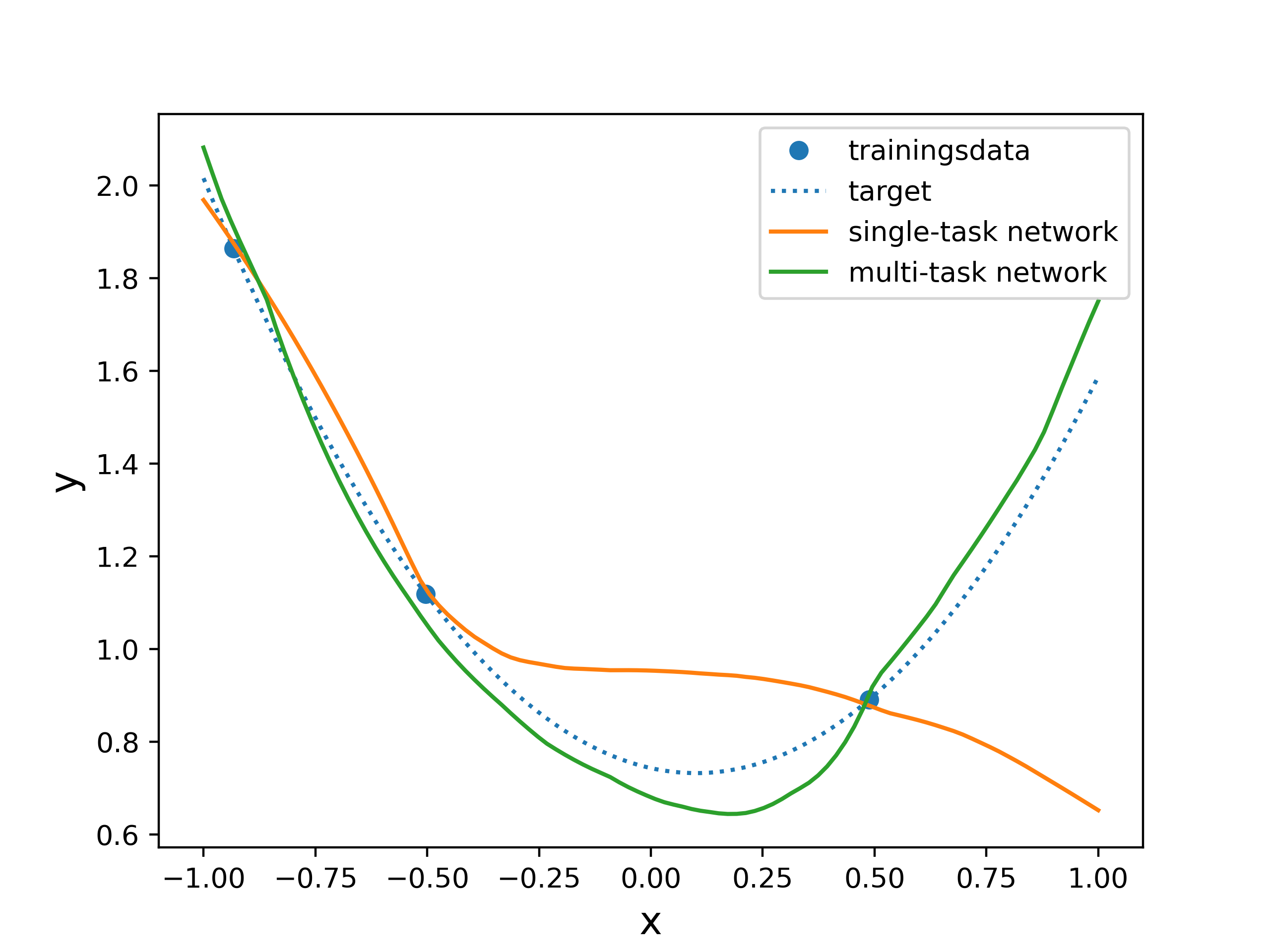}
\includegraphics[width=0.5\textwidth]{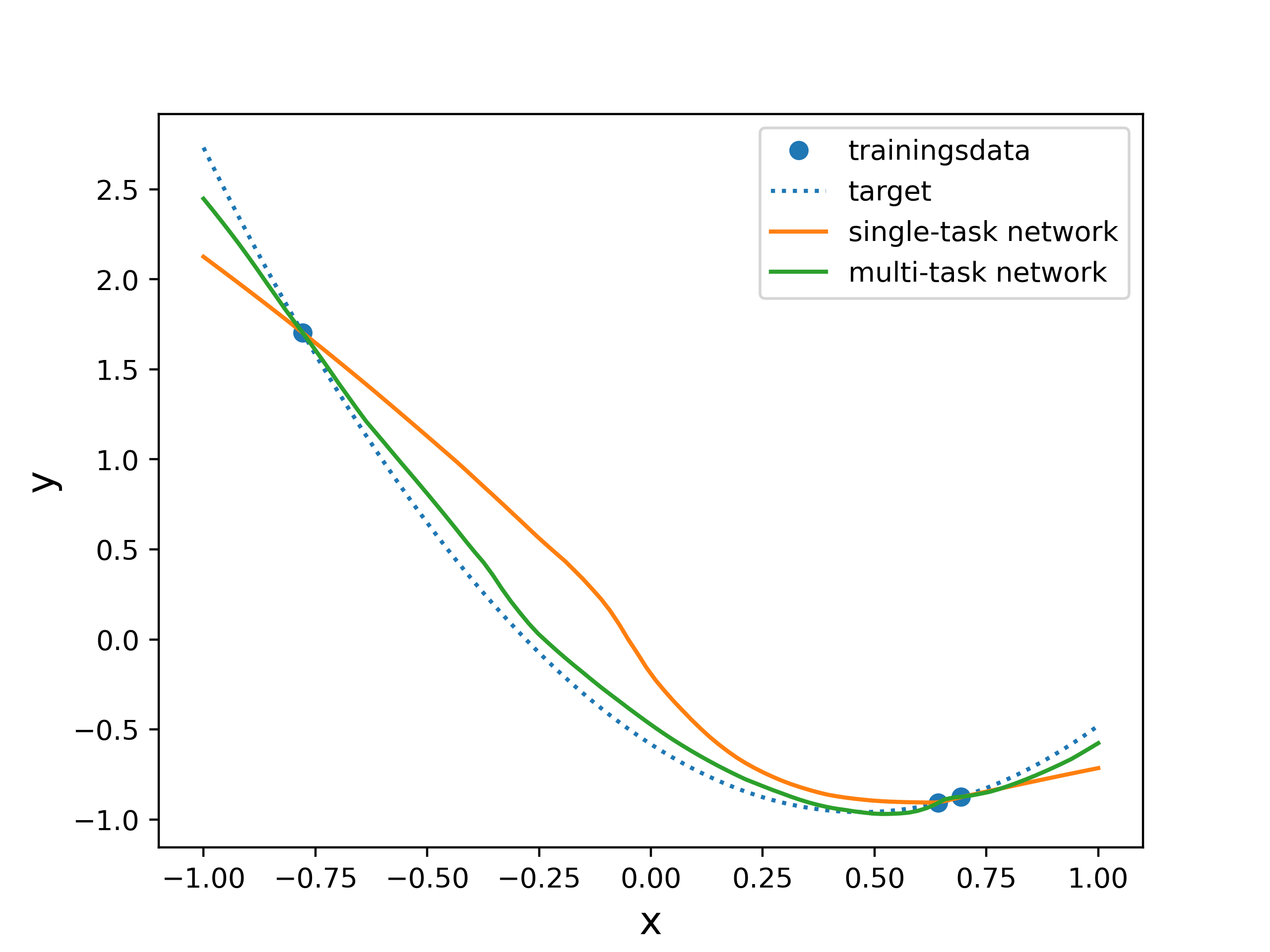}
\caption{Approximation functions from the first four tasks for a PNN with three parameters, compared to a neural network trained on the respective, single task only. The training data is shown by the dots.} \label{fig-quadratic-samples}
\end{figure}

\begin{figure}\centering
\includegraphics[width=0.45\textwidth]{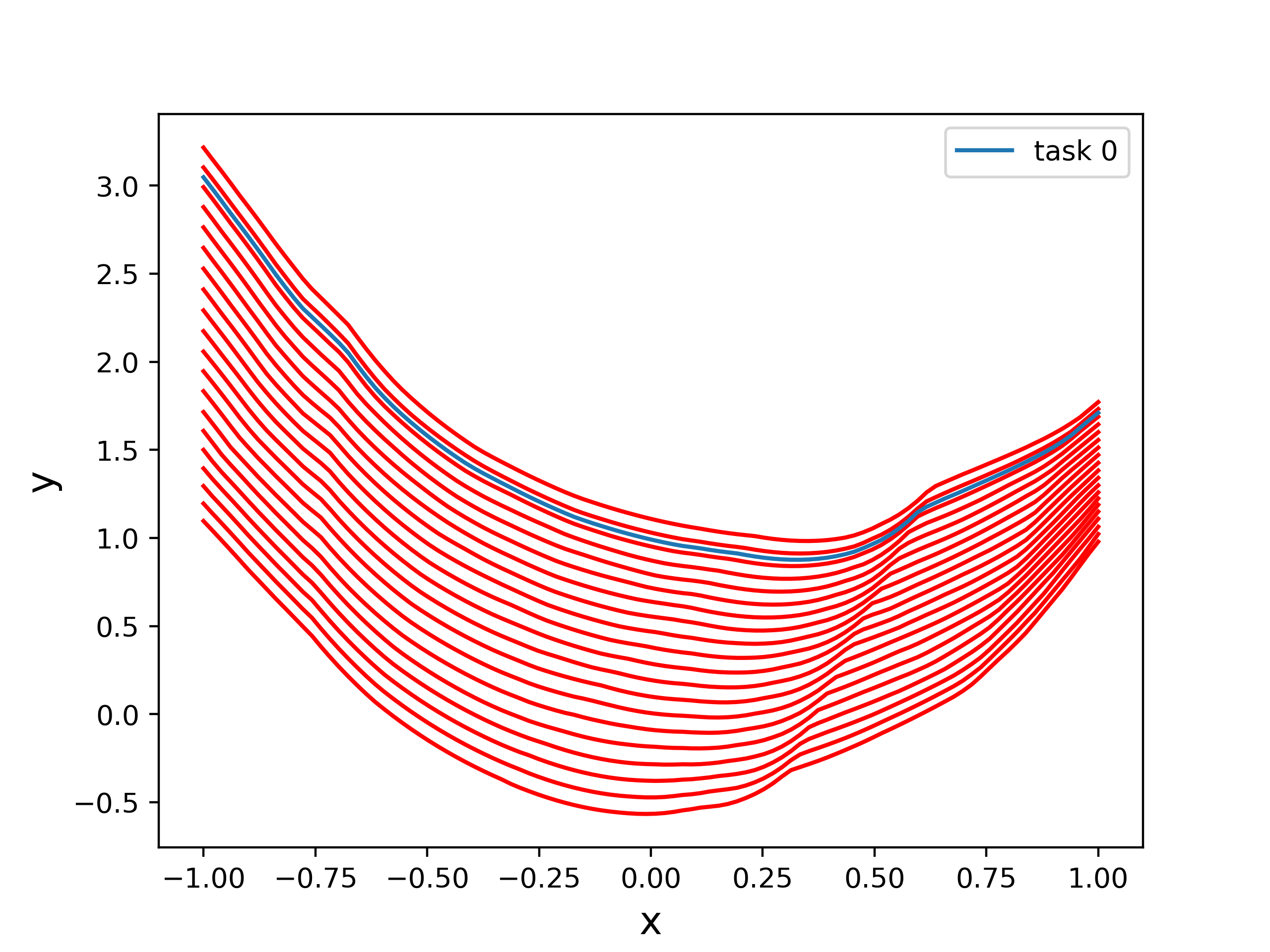}
\includegraphics[width=0.45\textwidth]{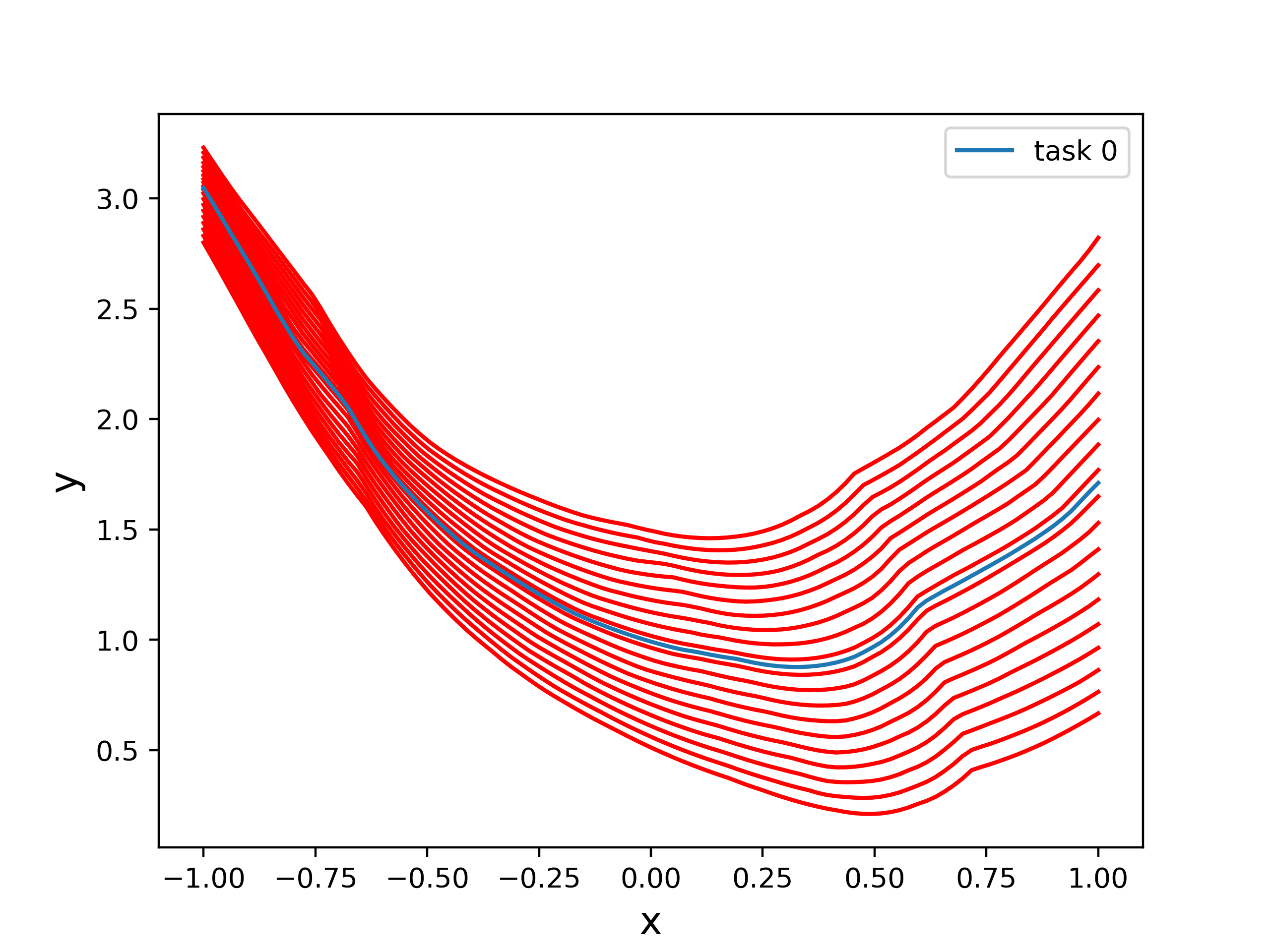}
\includegraphics[width=0.45\textwidth]{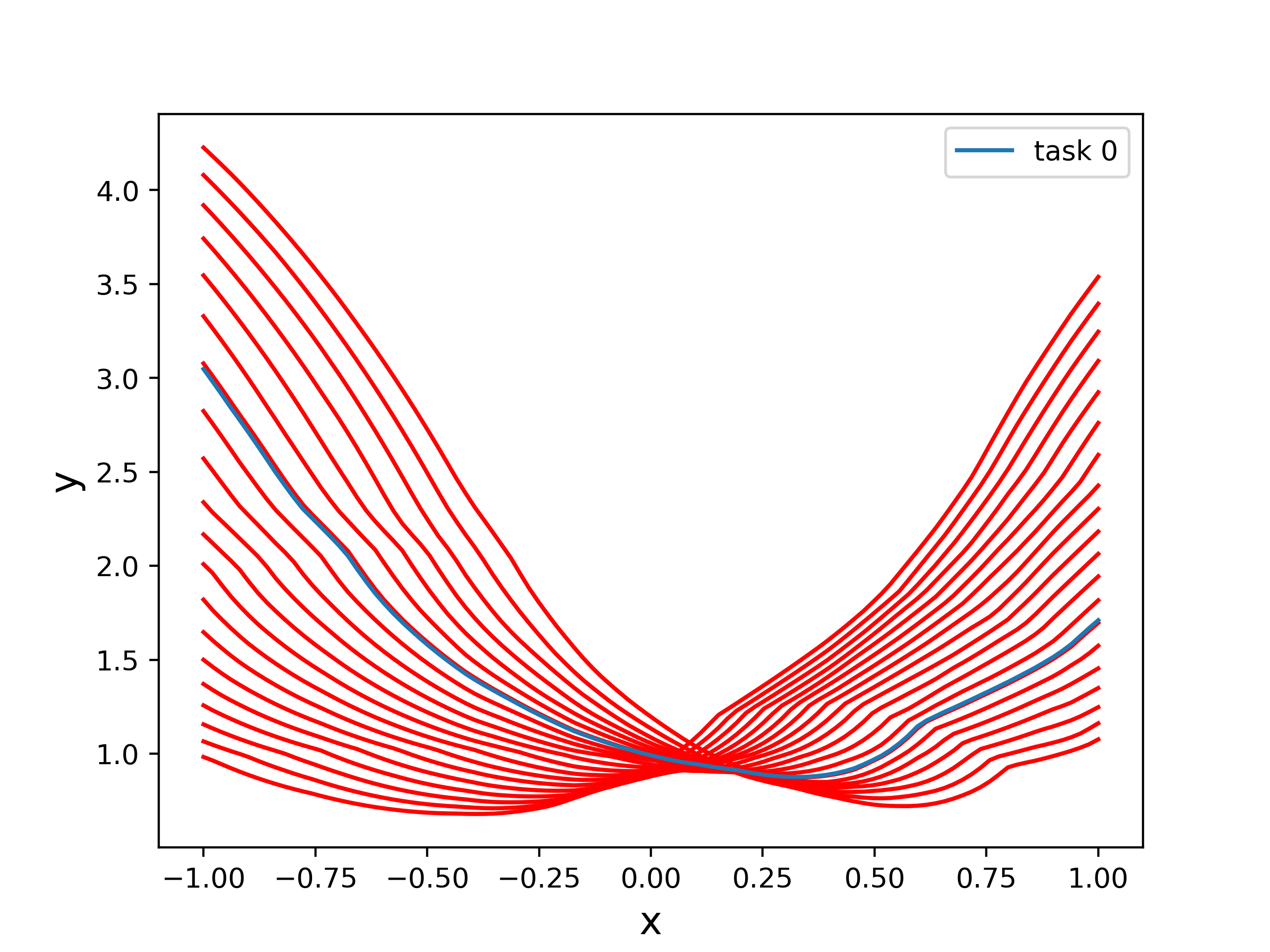}
\caption{Projection along a parameter direction (equidistant between the minimum and maximum of the respective parameter over all tasks)  where all other parameters are fixed.  } \label{fig-quadratic-projection}
\end{figure}
\begin{figure}\centering
\includegraphics[width=0.45\textwidth]{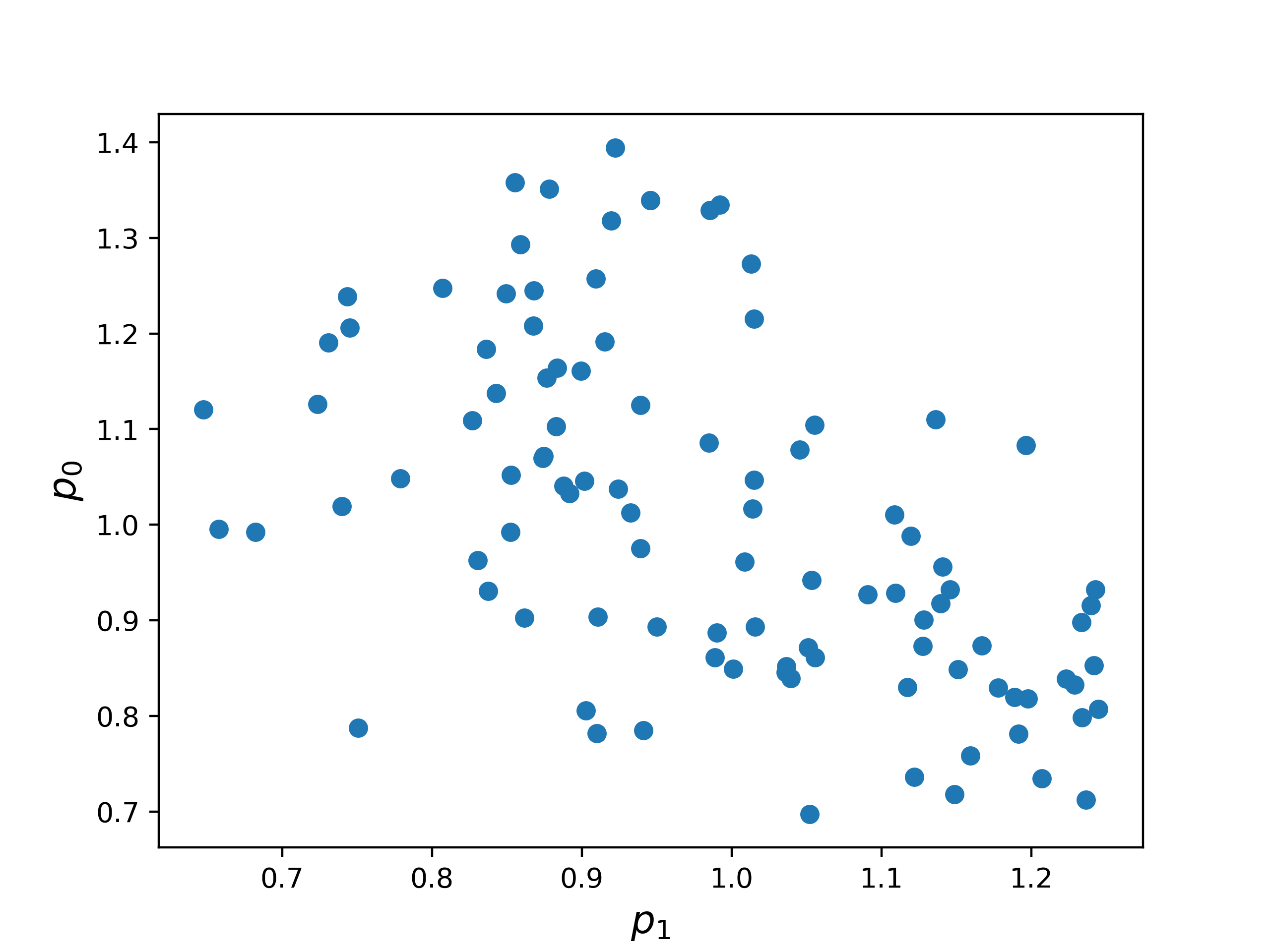}
\includegraphics[width=0.45\textwidth]{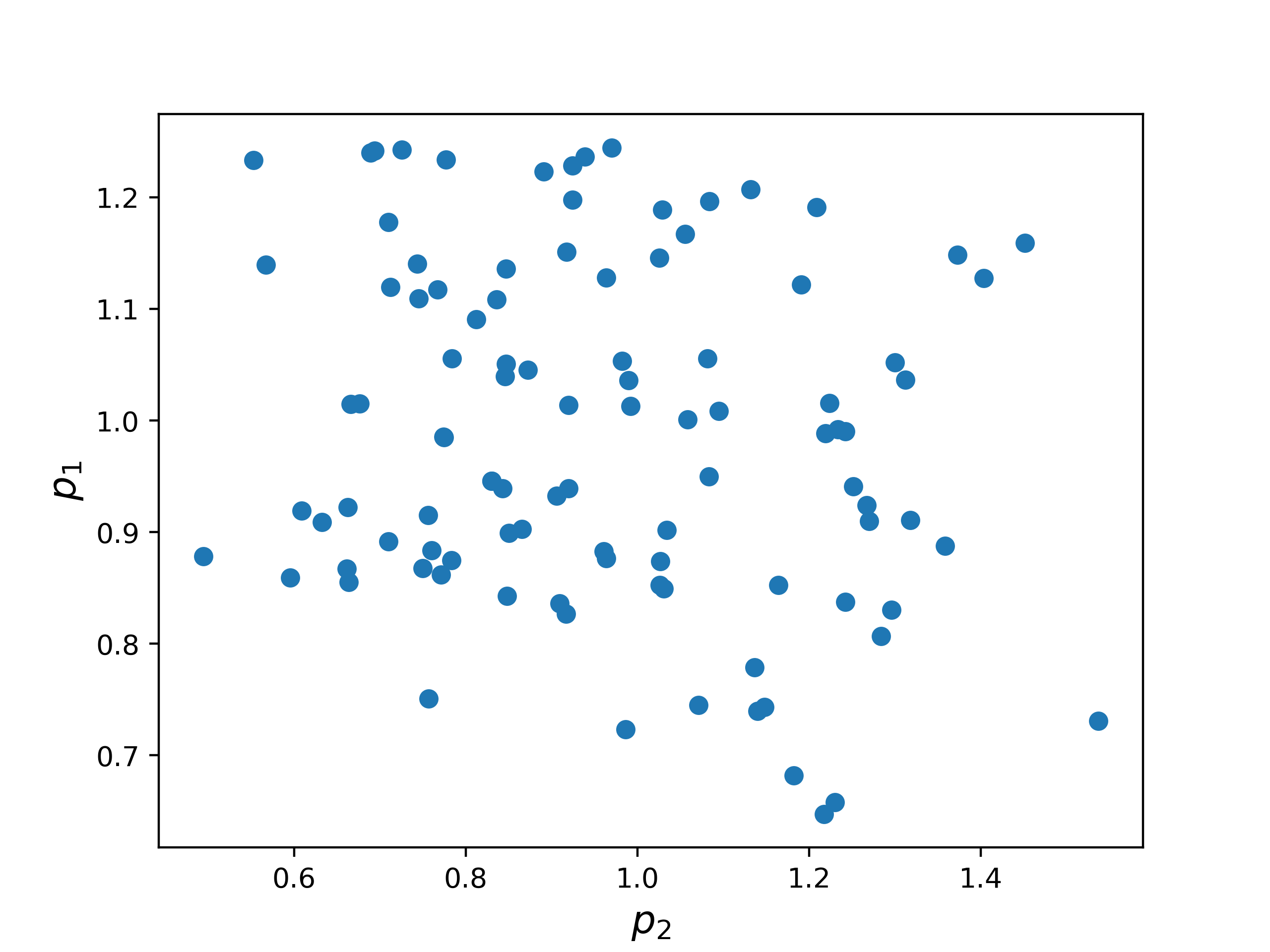}
\includegraphics[width=0.45\textwidth]{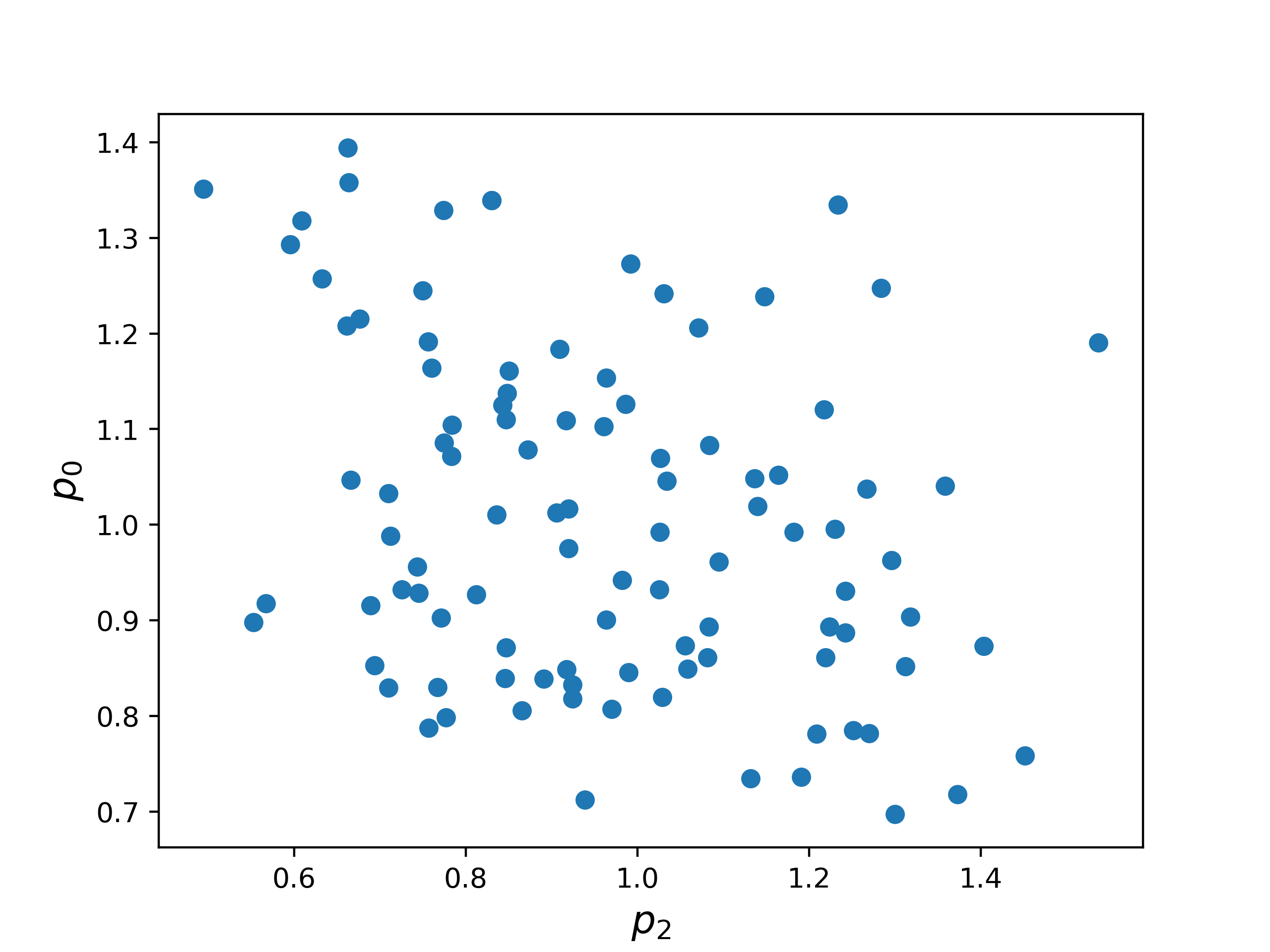}
\caption{Scatter plots for all parameters calibrated on the training set for the function family defined by (\ref{eq-quadratic}). } \label{fig-quadratic-scatter}
\end{figure}

In our first example, we consider a problem where each task is a simple quadratic function of the form
\begin{equation}\label{eq-quadratic}
	f(x;a,b,c) := a(x-c)^2+ b, \; x \in [-1,1],
\end{equation}
for parameters $a$, $b$, $c$.
For each task, we first sample the parameters from uniform distributions, such that $a\in [1.0,2.0]$, $b\in [-1,1]$, and $c\in [-0.5,0.5]$ to determine the function for this task. For each function we then sample $n$ points, uniformly distributed on $[-1,1]$, and evaluate the function  on these  points.  
Figure \ref{fig-quadratic} shows samples and errors for the case with three points per task. In this experiment, we analyze the performance of the method in relation to the number of parameters used in the PNN.  
Since the family of functions has three parameters, we expect that this should also be the optimal number for the parameter dimension our PNN. 

The right graph in figure \ref{fig-quadratic} shows the error corresponding to the number of parameters as well as the mean of the error for eight different calibrations using different random seeds for network initialization. As a baseline, we also plot the mean error calibrating a neural network on each task separately. We clearly see that PNN outperforms the calibration of networks for each task separately. Furthermore, for one and two parameters the error is slightly larger than for a higher number of parameters. This is not surprising, considering that the generating family of functions depends on three parameters. We also see that an increasing number of parameters does not affect the performance of the resulting network significantly.
Figure \ref{fig-quadratic-samples} shows some examples of the resulting PNN approximation functions compared to a simple feed-forward network for four selected tasks. We see that the PNN has learned the parabolic shape of the target problem much better than the single-task network. The upper right graph shows that even extrapolation for points between -1 and 0  gives quite satisfactory results.

Table \ref{tab-quadratic} shows the error results of a PNN depending on the number of data points per task, as well as on the number of tasks, compared to the error results of neural networks fitted to each task separately. We clearly see that the number of points from the training data as well as the number of tasks influences the overall approximation error.  The error decreases by a factor of two with an increasing number of points per task, for the PNN as well as for the single networks. As indicated by theorem \ref{Generalization-Property}, we observe a similar effect on the error with an increasing number of tasks. 
Note that performance seems to slightly deteriorate from 50 to 100 tasks.
One reason for this effect might be that we used the same training parameters (number of epochs, learning rate schedule) across all configurations without tuning these parameters individually for each number of tasks leading and the networks have not fully converged to the desired accuracy.
\begin{table}
\centering
\begin{tabular}{ccrrrr}
\hline
number               &  simple & \multicolumn{4}{c}{number tasks} \\
points               &  network &    10 &    20 &    50 &   100 \\
\hline
3                              &   0.669 & 0.273 & 0.236 & 0.176 & 0.163 \\
4                              &   0.538 & 0.223 & 0.258 & 0.137 & 0.141 \\
5                              &   0.481 & 0.243 & 0.117 & 0.137 & 0.115 \\
6                              &   0.292 & 0.134 & 0.113 & 0.108 & 0.092 \\
\hline
\end{tabular}
\caption{Test error for a family of 250 quadratic functions as defined in (\ref{eq-quadratic})  for a network fitted to each task separately (column \emph{simple metwork}) and for a PNN with three parameters, trained on  different numbers of tasks and points per sample.} \label{tab-quadratic}
\end{table}
The role of the parameters is shown in figure \ref{fig-quadratic-projection}, where projections along each parameter dimension are shown. We use the parameter from task zero as a basis and vary the parameter coordinate between the minimum and maximum of the training data set in each figure. The behavior of the PNN with respect to each parameter coordinate is quite different and also mutually independent between the coordinates. This observation is confirmed by the scatter plots of the parameters in figure \ref{fig-quadratic-scatter}, where the distribution of the parameters seems to be rather uncorrelated and uniformly distributed. Furthermore, the change of a single parameter coordinate leads to a parabolic-shaped curve.

\subsection{Family of quadratic functions with noise}

\begin{figure}
\includegraphics[width=0.5\textwidth]{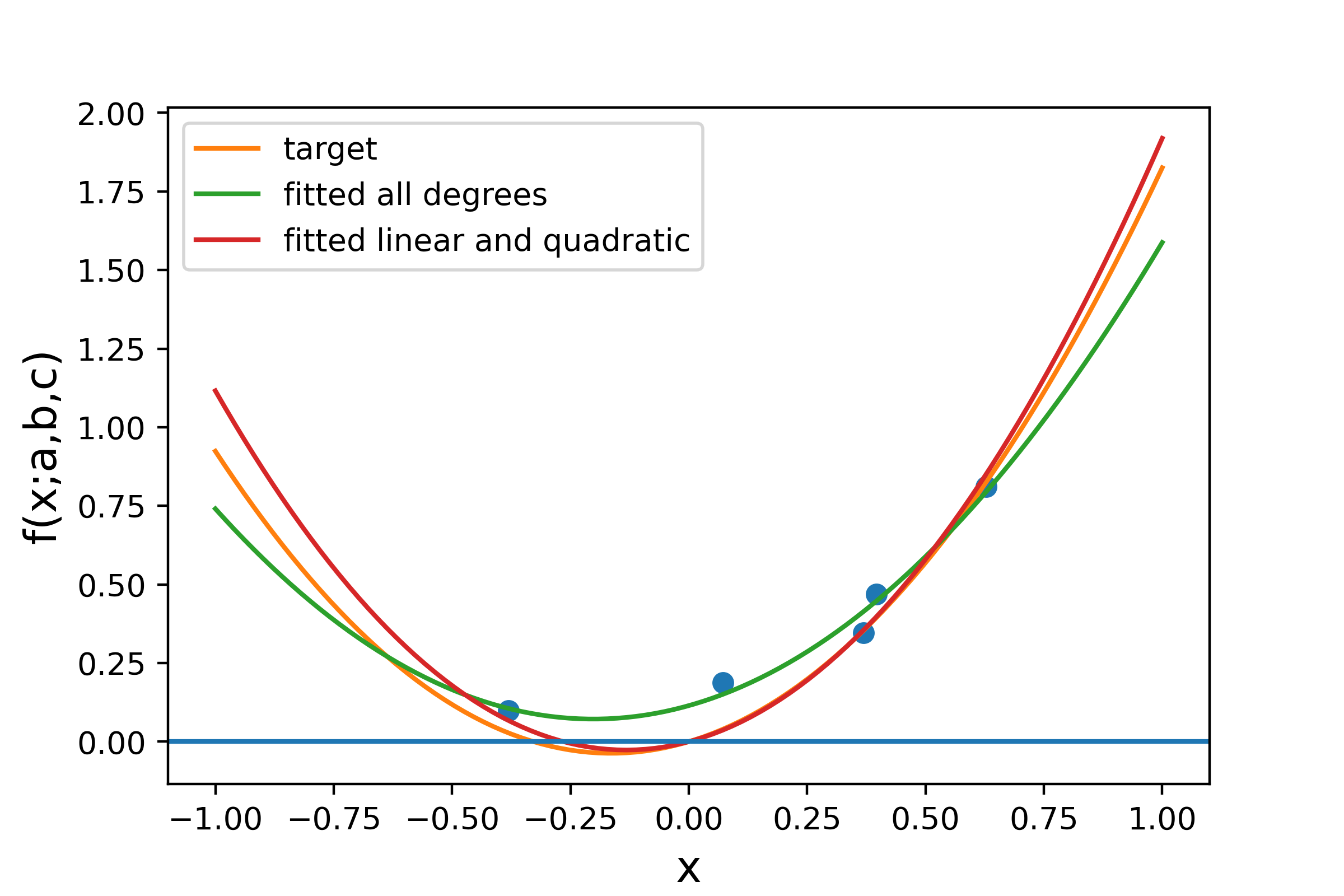}
\includegraphics[width=0.5\textwidth]{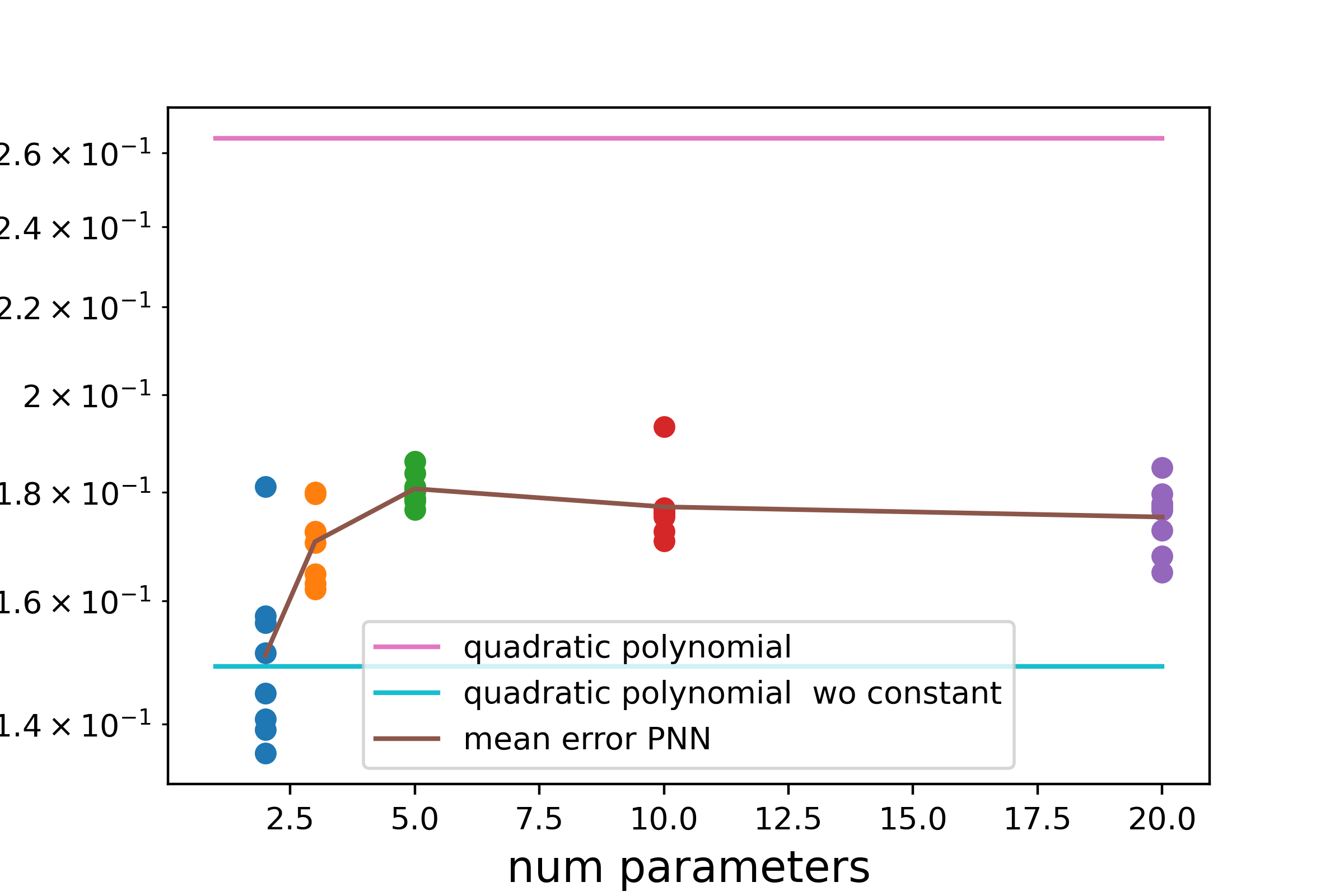}
\caption{Left: A task sampled from the family defined by (\ref{eq-quadratic-noise}) together with the generating function (before noise is added) and the regression using the quadratic function and the quadratic function with zero constant.  Right: The error of the PNN for different numbers of parameters. As baselines, the error of the simple regressions (quadratic and quadratic with constant term fixed to zero) are plotted as straight lines. } \label{fig-quadratic-noise}
\end{figure}
\begin{figure}\centering
\includegraphics[width=0.7\textwidth]{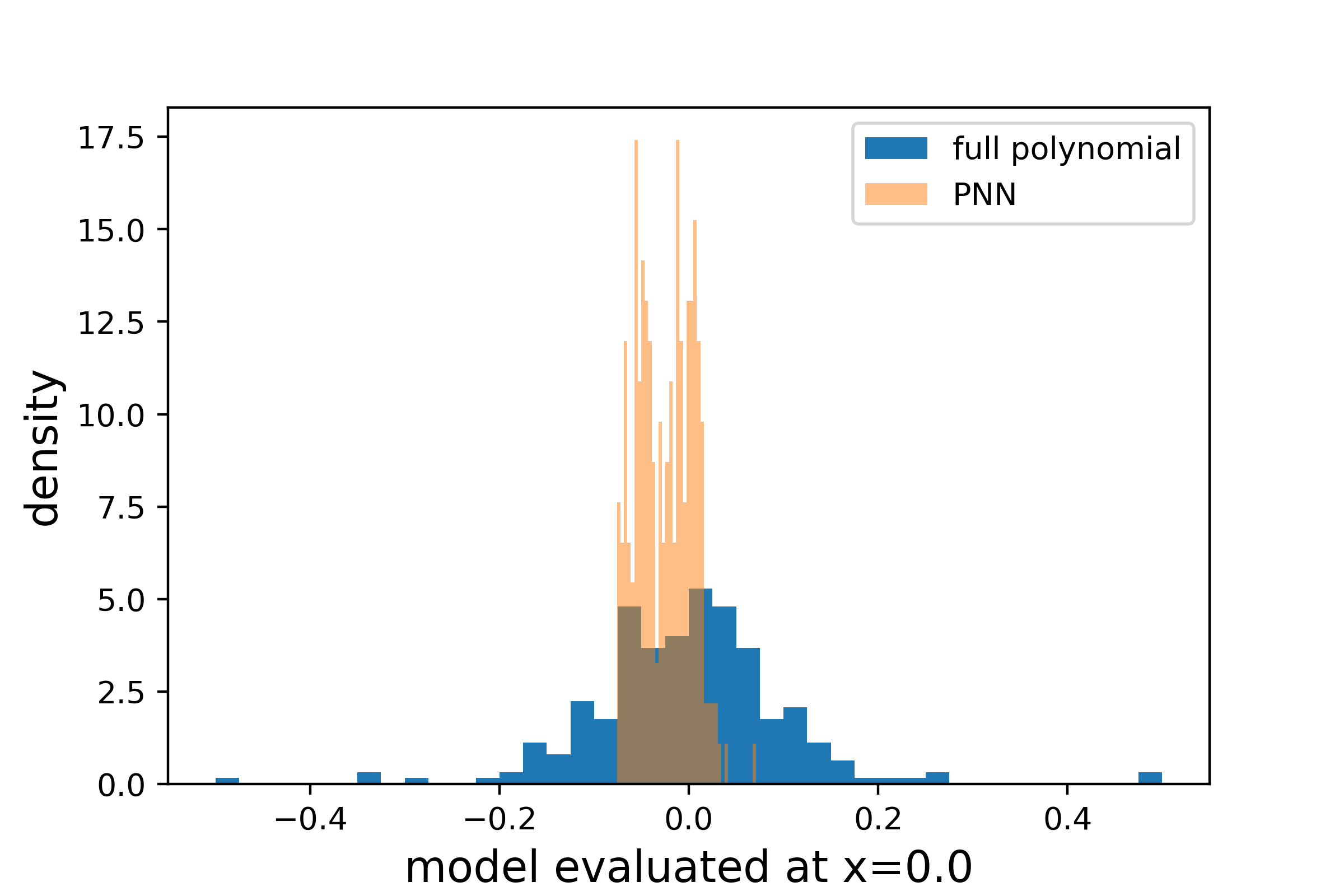}
\caption{ Distribution of the predicted values at $x=0$ from a PNN fitted to 250 tasks sampled from the family of functions defined by (\ref{eq-quadratic-noise}), as well as from the respective quadratic polynomial regressions.} \label{fig-quadratic-zero}
\end{figure}
\begin{figure}\centering
\includegraphics[width=0.49\textwidth]{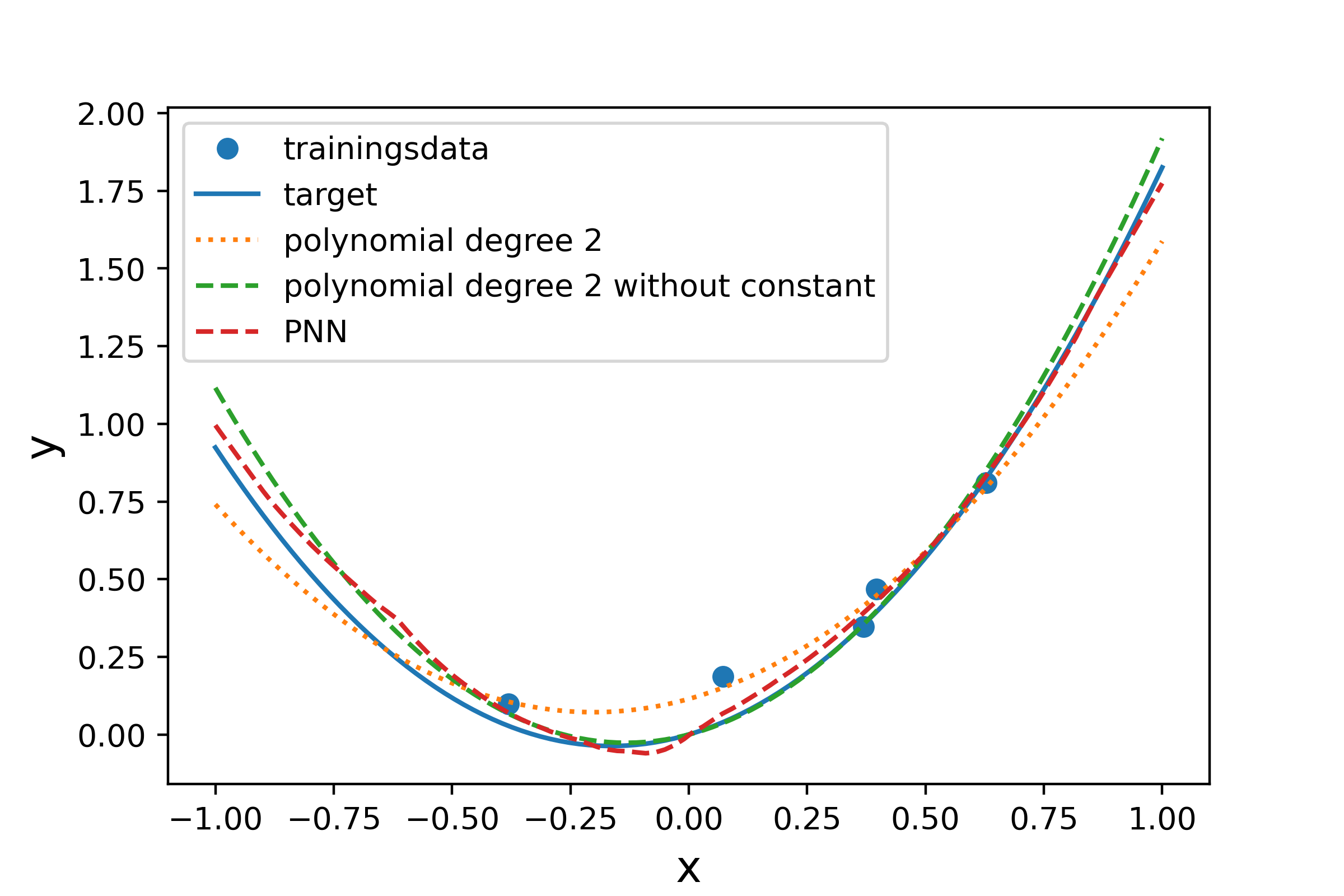}
\includegraphics[width=0.49\textwidth]{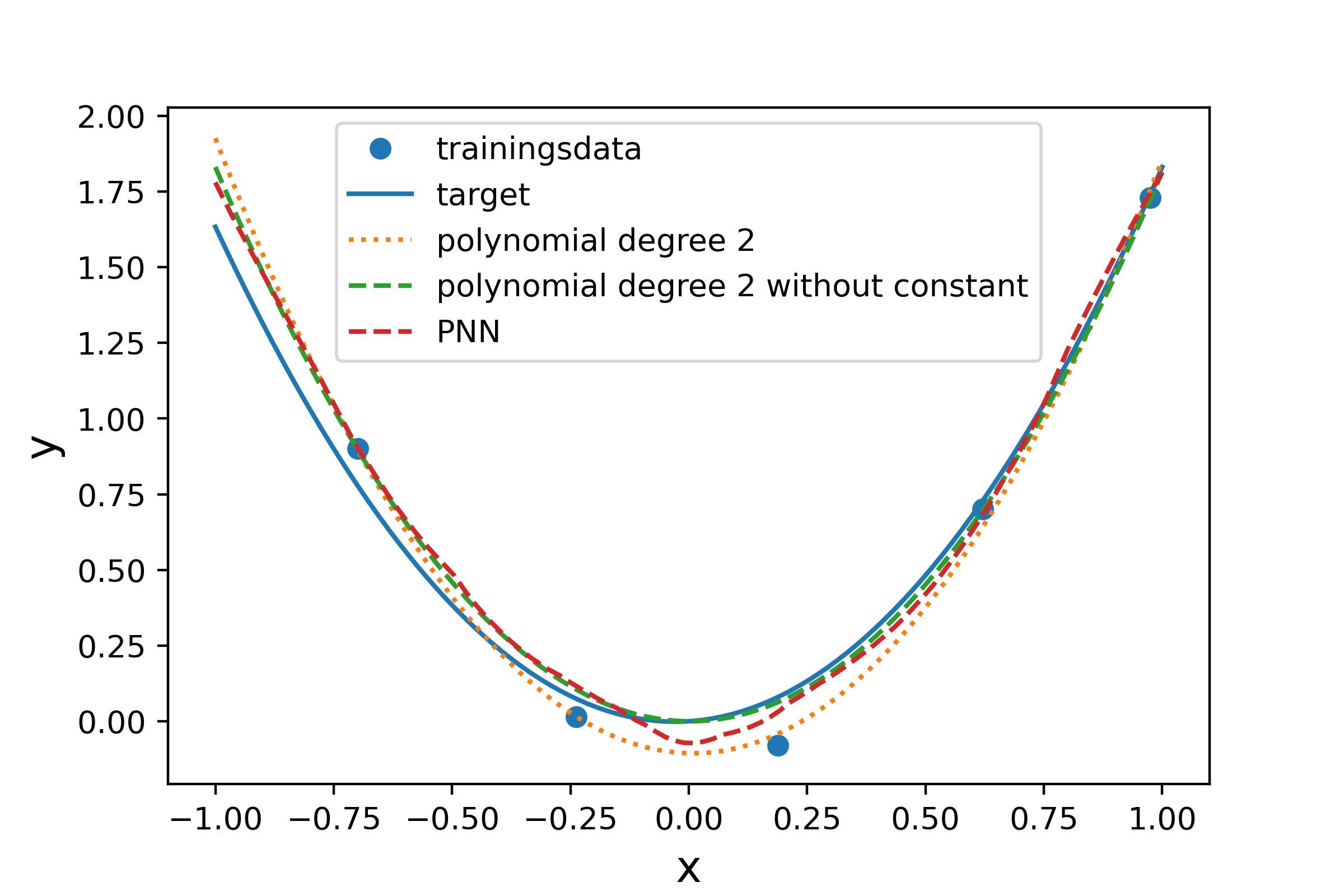}
\caption{Function values from PNN, polynomial regression, and target function for two tasks sampled by (\ref{eq-quadratic-noise}).} \label{fig-quadratic-zero-functions}
\end{figure}

We now consider the case of noisy data, again, generated by a family of quadratic functions, but with a bit more structure than in the previous example. Let
\begin{equation}\label{eq-quadratic-noise}
	f(x;a,b) := ax^2+bx +\varepsilon, \; x \in [-1,1],
\end{equation}
where $\varepsilon$ is normally distributed with standard deviation $0.1$. The parameters $a\in [1.0,2.0]$ and $b\in [-0.5,0.5]$ are uniformly sampled. Note that for any function of this family $f(0; a,b)=0$ is true. Due to the noise term, we sample five points per task for building the training data inputs.
In the following we compare the PNN results with two benchmarks based on quadratic polynomial regression: the first one including an estimation of the constant, and the second one setting the constant term to zero and estimating only the linear and quadratic coefficients. Ignoring the bounds for the parameters $a$ and $b$  the quadratic regression with zero constant seems to be the best possible model class for this kind of data. The right graph in figure \ref{fig-quadratic-noise} shows the error on the training data for an increasing number of parameters and eight different networks (with different initial weights), as well as the results for the two regression models. Note that the error is measured between the PNN and the target function \emph{without} noise. Independently of the number of parameters, the PNN provides smaller errors than the regression model using quadratic polynomials. Moreover, the mean error for two parameters is nearly equal to the error using the quadratic regression with  zero constant. In contrast to the previous example, we see that the errors for PNNs with parameter dimension greater than two are slightly higher than for PNNs with two parameters. The reason might be that due to the noise term a bit of overfitting is introduced.  Using some kind of regularization might further improve the results. However, even for larger number of parameters the results are quite good, having in mind that the performance is still better than applying a quadratic polynomial regression. From these results we see, that the model is able to learn the functional structure on noisy data too. Moreover, since the results for the PNN are better than for the quadratic model (with fitted constant), we can assume that the PNN is able to learn the property of the real function being equal to zero for $x=0$. This is visually confirmed by figure \ref{fig-quadratic-zero}, where the distribution of the function values at $x=0$ are plotted for the PNN and for the polynomial regression. Figure \ref{fig-quadratic-zero-functions} shows the target function, the PNN regression, and the polynomial regression functions for three different tasks.

\subsection{Family of quadratic functions with Interdependencies}

\begin{figure}
\includegraphics[width=0.5\textwidth]{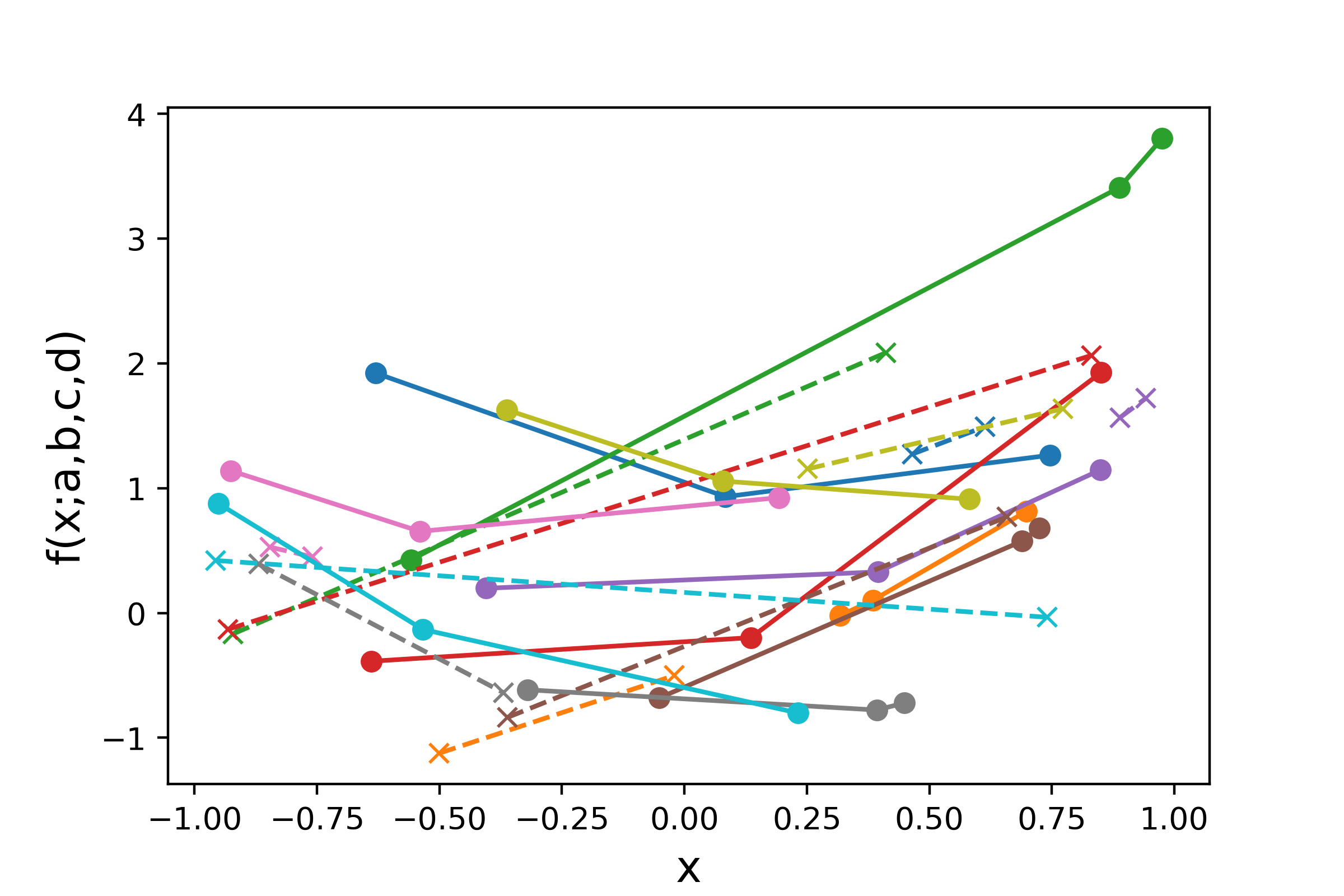}
\includegraphics[width=0.5\textwidth]{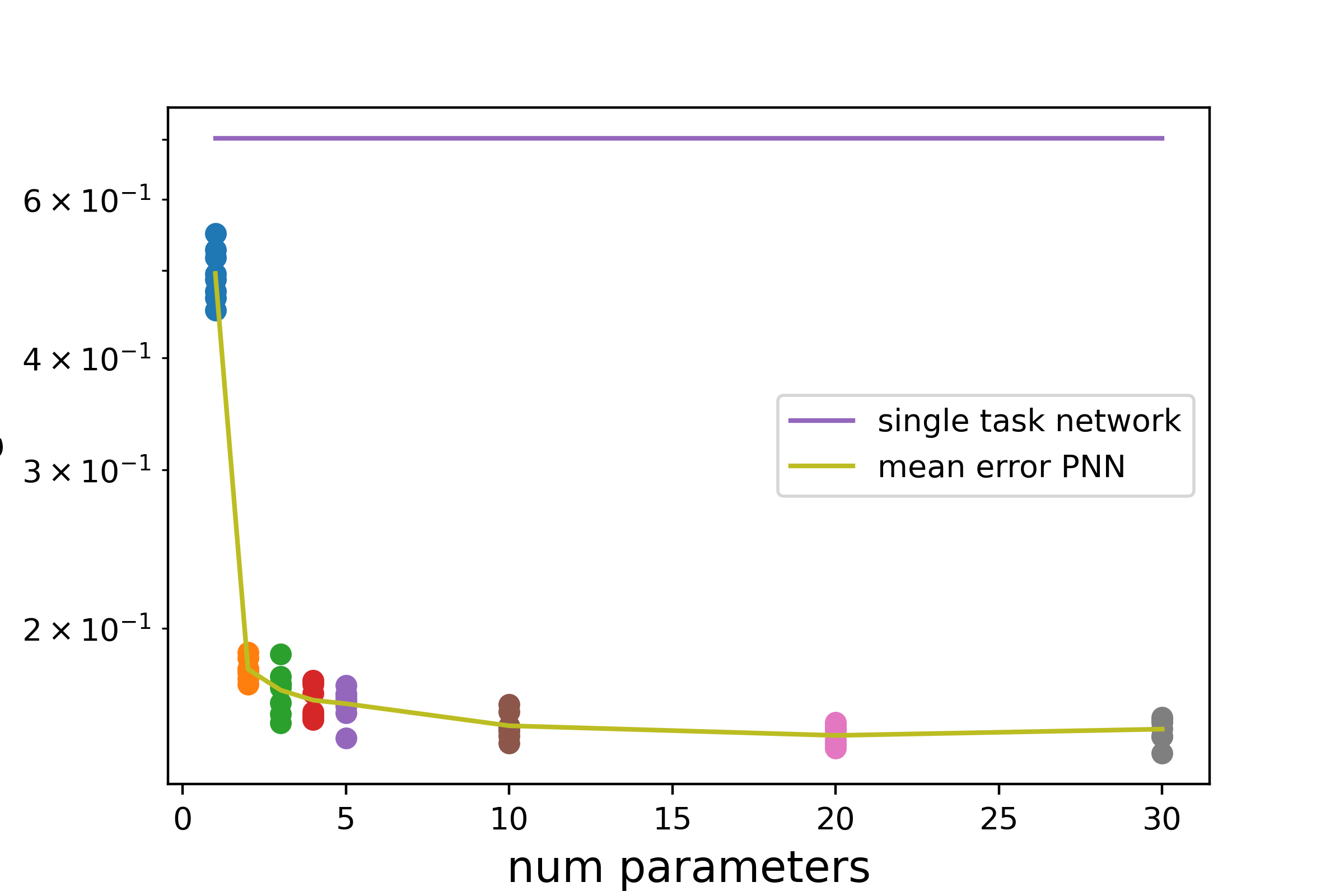}
\caption{Left: Several tasks sampled from (\ref{eq-quadratic2}). The straight lines correspond to the case $x_2=0$ and the dashed lines to $x_2=1$. Right: Approximation errors for different numbers of parameters of the PNN, using 100 different tasks for training.} \label{fig-quadratic2}
\end{figure}
\begin{figure}
\includegraphics[width=0.5\textwidth]{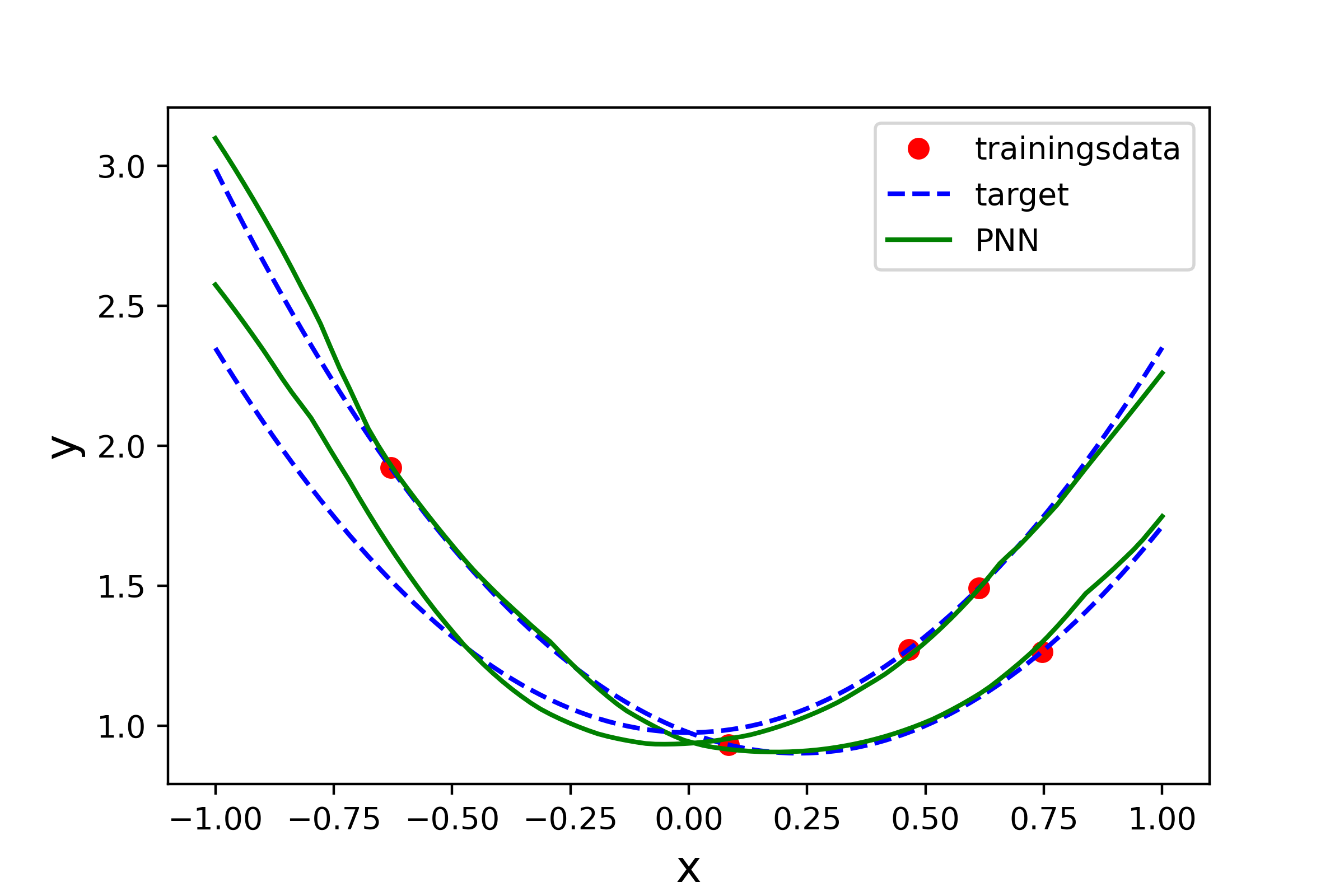}
\includegraphics[width=0.5\textwidth]{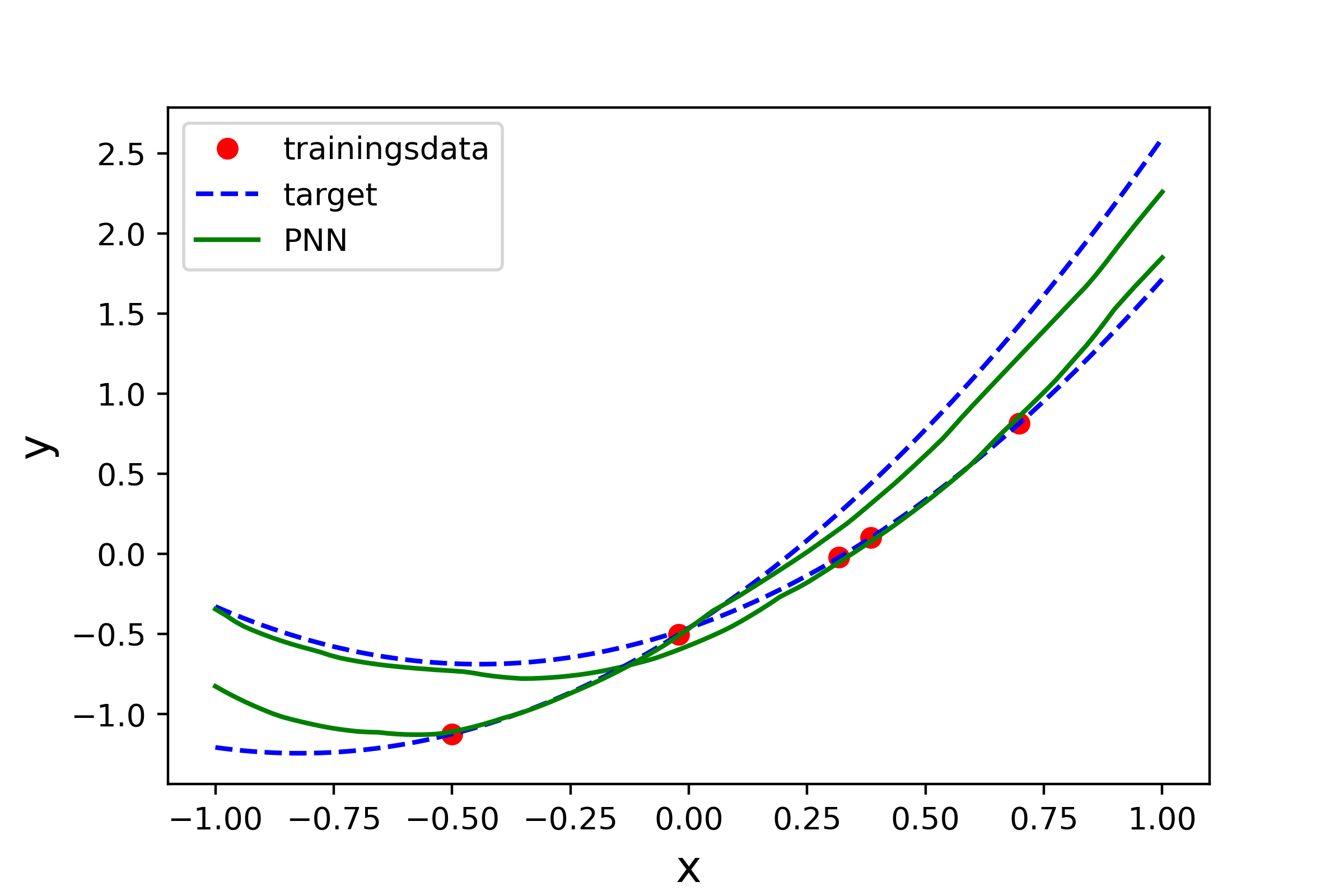}
\includegraphics[width=0.5\textwidth]{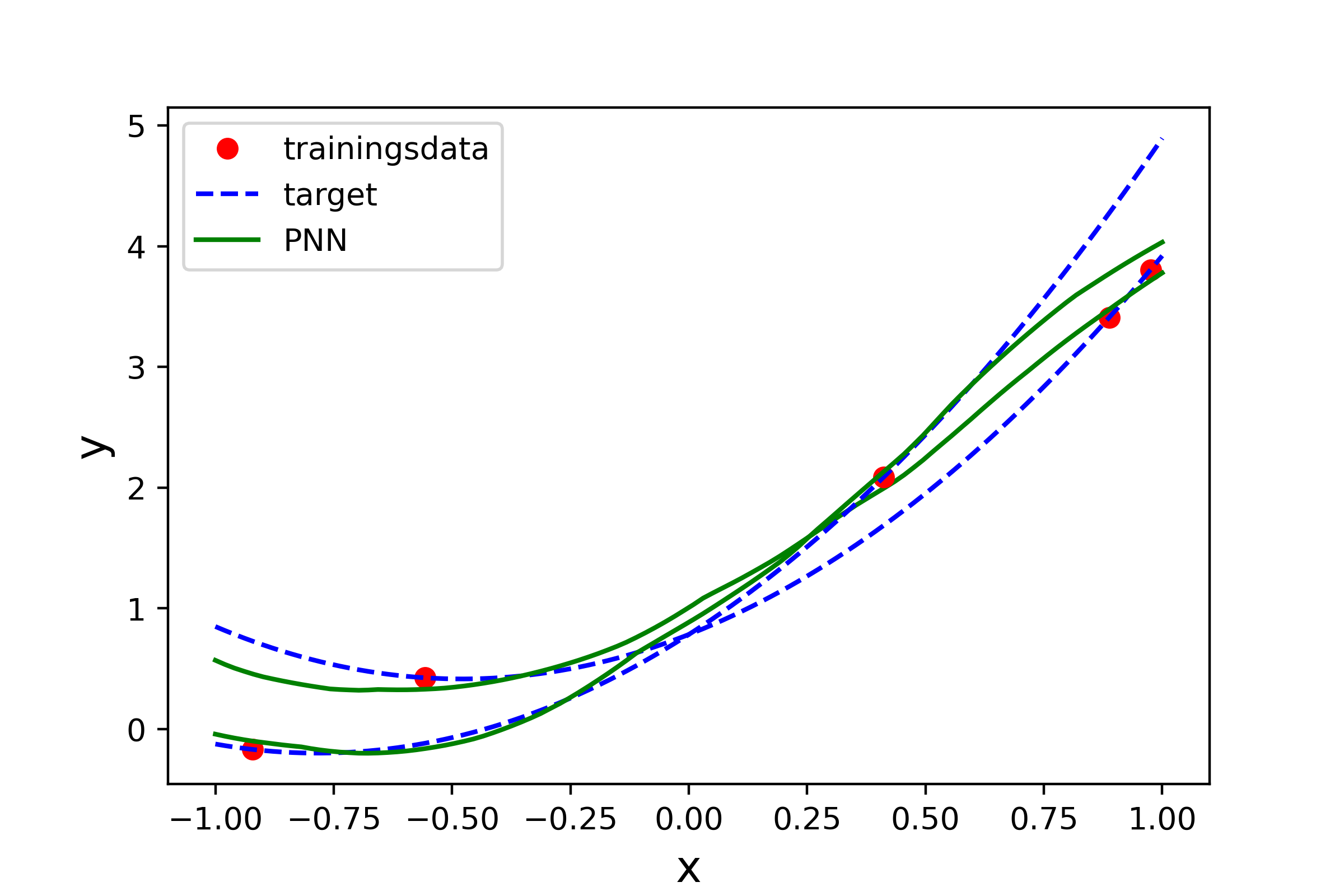}
\includegraphics[width=0.5\textwidth]{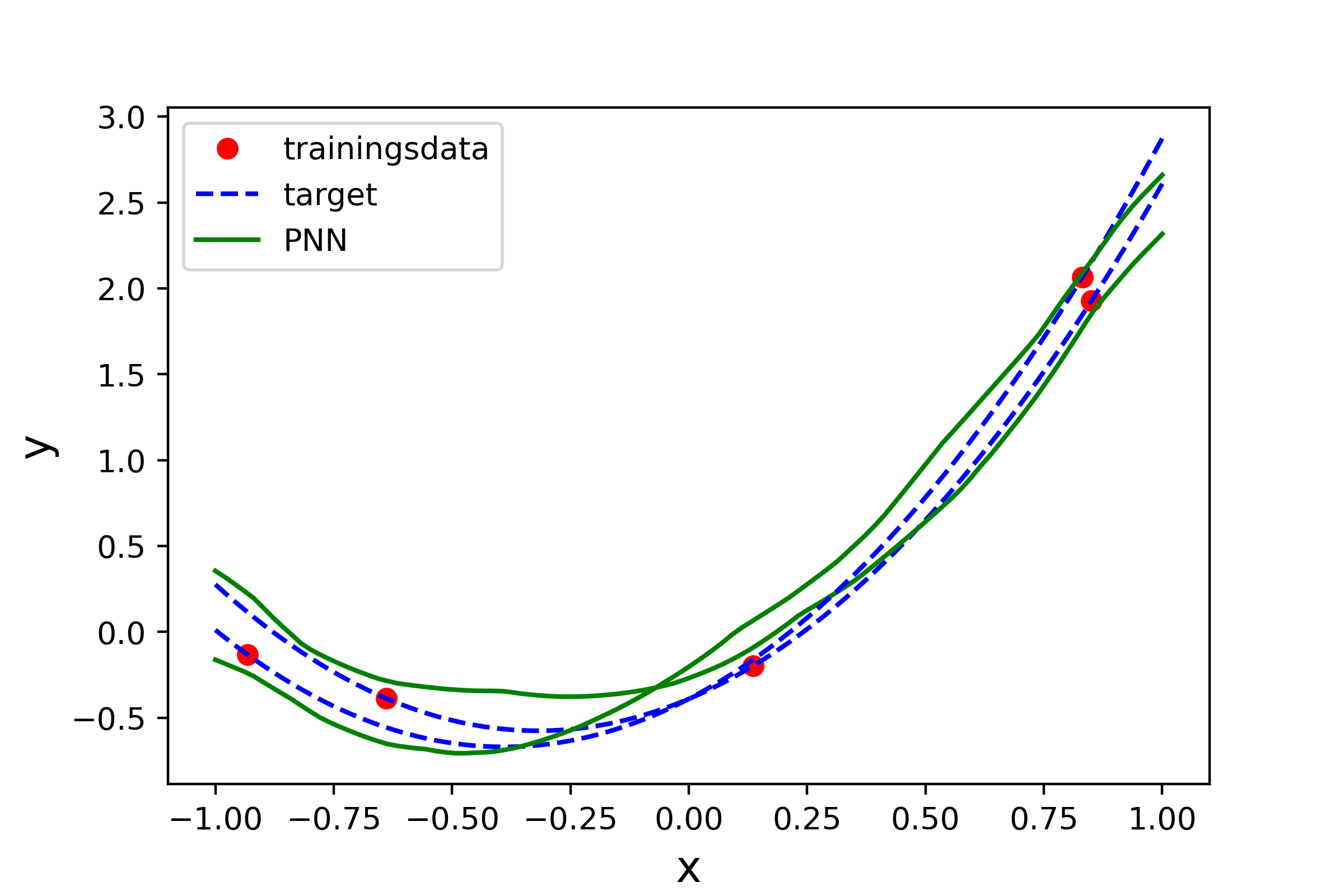}
\caption{Four tasks sampled from (\ref{eq-quadratic2}) with true functions and PNN estimations.} \label{fig-quadratic2-functions}
\end{figure}
Many financial applications involve binary or categorical features such as ratings or countries. As an example, consider the interest rate spread curves mentioned before.
With real data it happens often that input data is unbalanced in the sense that some categories occur much lesser than others. For instance, for a developed interest rate market one will find a large amount of bond prices for all rating categories, whereas smaller country may have only liquid prices for some of the ratings.

In such cases, PNNs may help to learn relationships between categories to improve results for underrepresented data. To analyze such behavior we perform the following simple experiment. 

For $x$ in $[-1,1]\times \{0,1\}$ we define the function family
\begin{equation}\label{eq-quadratic2}
	f(x;a,b,c,d) := \left\{ \begin{array}{ll}  a(x_1-c)^2+b + dx_1 & \mbox{ if } x_2=1, \\
                                                a(x_1-c)^2+b & \mbox{ otherwise.}
                                                \end{array}\right.
\end{equation}
As before, each task is constructed by uniformly sampling $a\in [1.0,2.0]$, $b\in [-1,1]$, $c\in [-0.5,0.5]$ and $d\in [0.1,1.0]$. We generate five $x$ values per task, where $x_1$ is uniformly sampled from $[-1,1]$, and for three of these five samples $x_2$ is set to $0$, and $x_2=1$ for the other two. Note, that splitting each task into two separate estimation problems (according to the binary value $x_2$) does not work, since in the case $x_2 = 1$ only two data points per task are given, which is not enough to recover the underlying quadratic structure of the function from the data. Several sample points of different tasks are plotted in the left graph of figure \ref{fig-quadratic2}. The right graph in figure \ref{fig-quadratic2} shows the error for different parameter dimensions together with the errors of a simple neural network fitted to each task separately. As in the previous examples, we see a significant improvement for parameter dimensions greater than 1 compared to the simple neural network case. The PNN learned from the tasks that the true function for $x_2=1$ is of parabolic shaped too. In figure \ref{fig-quadratic2-functions} the true and the estimated functions are plotted for four selected tasks.

\subsection{Regimes}
\begin{figure}
\includegraphics[width=0.5\textwidth]{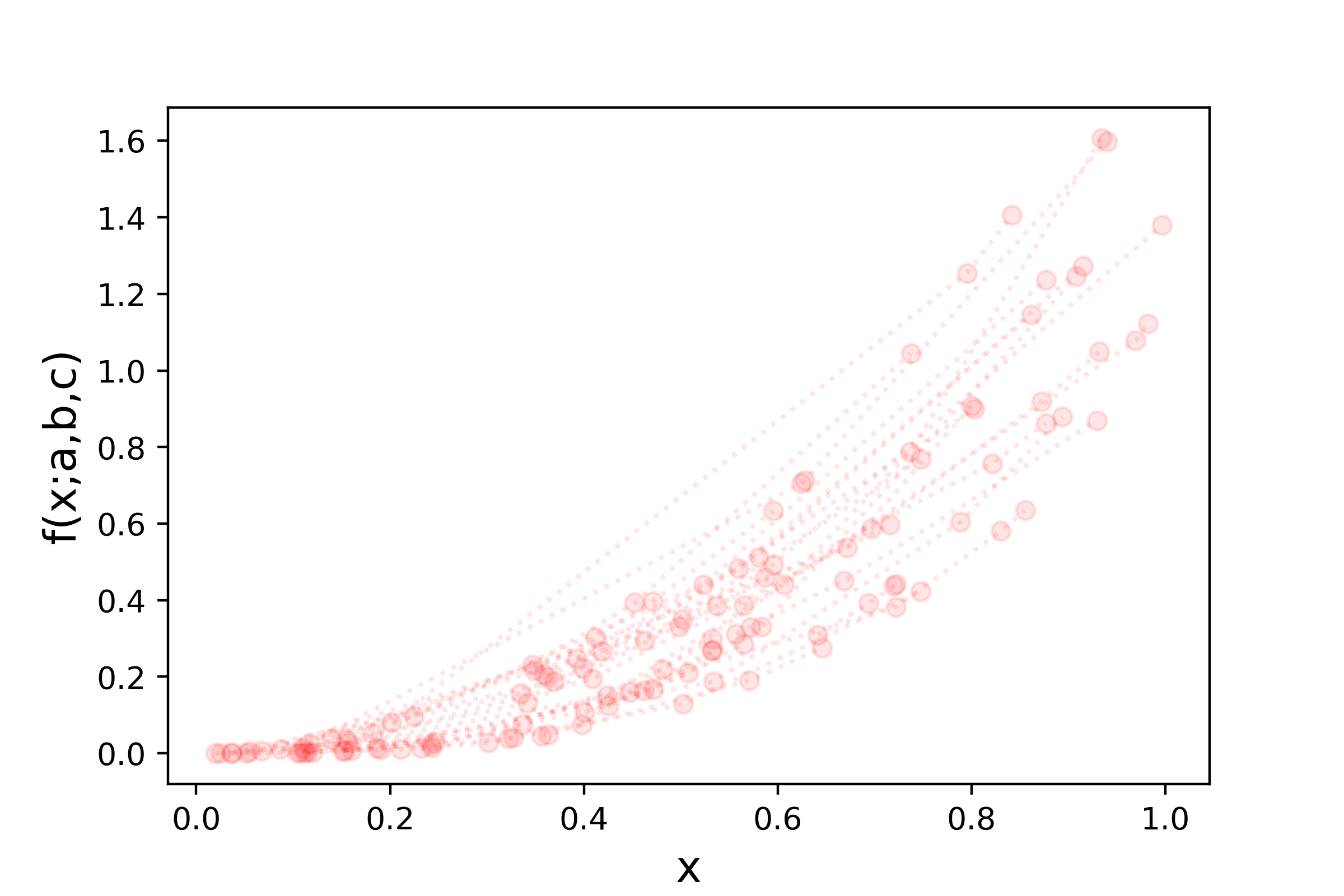}
\includegraphics[width=0.5\textwidth]{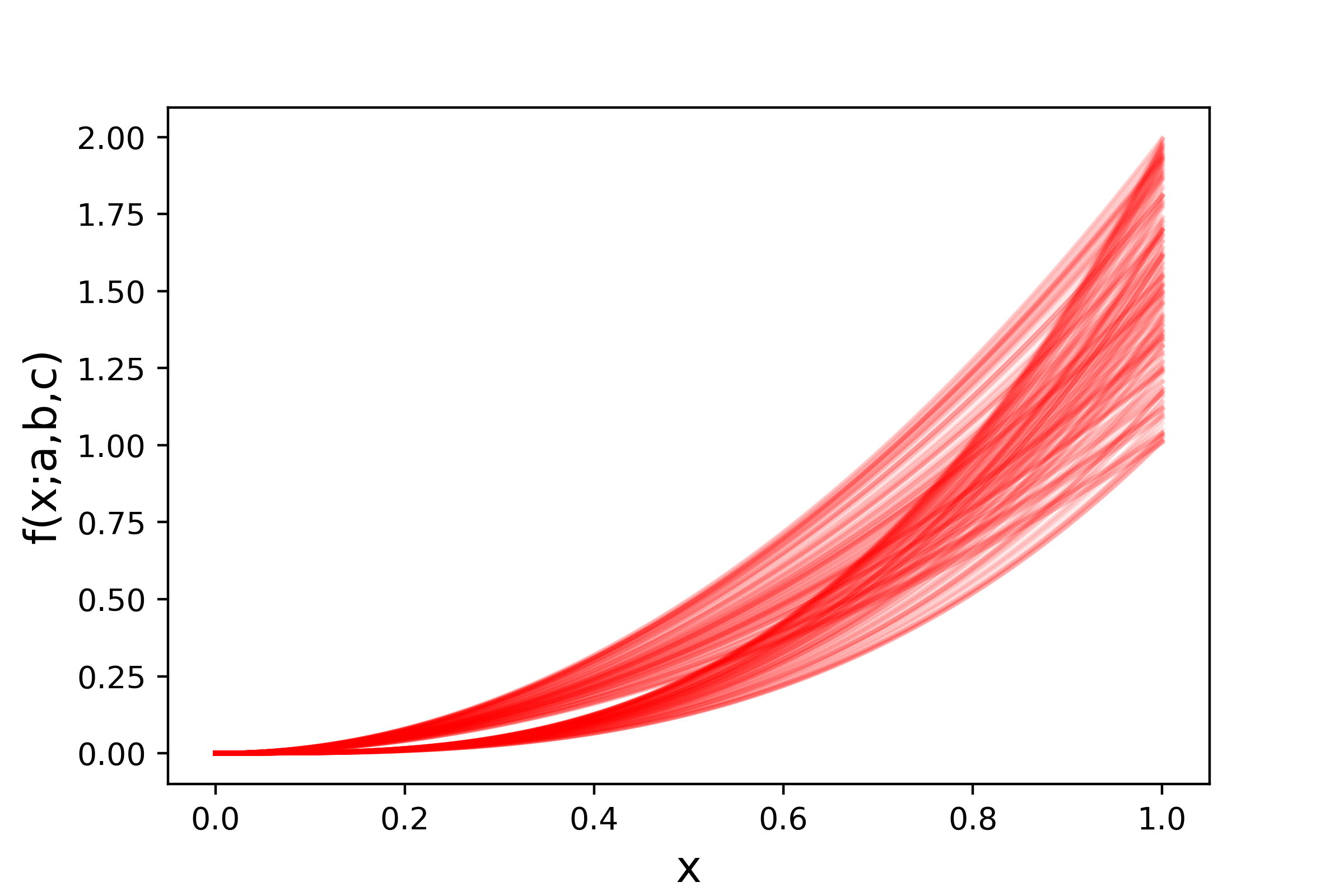}
\caption{Left: Sample data for 30 tasks for the function family in (\ref{fig-tasks-regime}), interpolated by straight lines for each sampled function. Right: 250 sampled functions (right).} \label{fig-regime}
\end{figure}
\begin{figure}\centering
\includegraphics[width=0.49\textwidth]{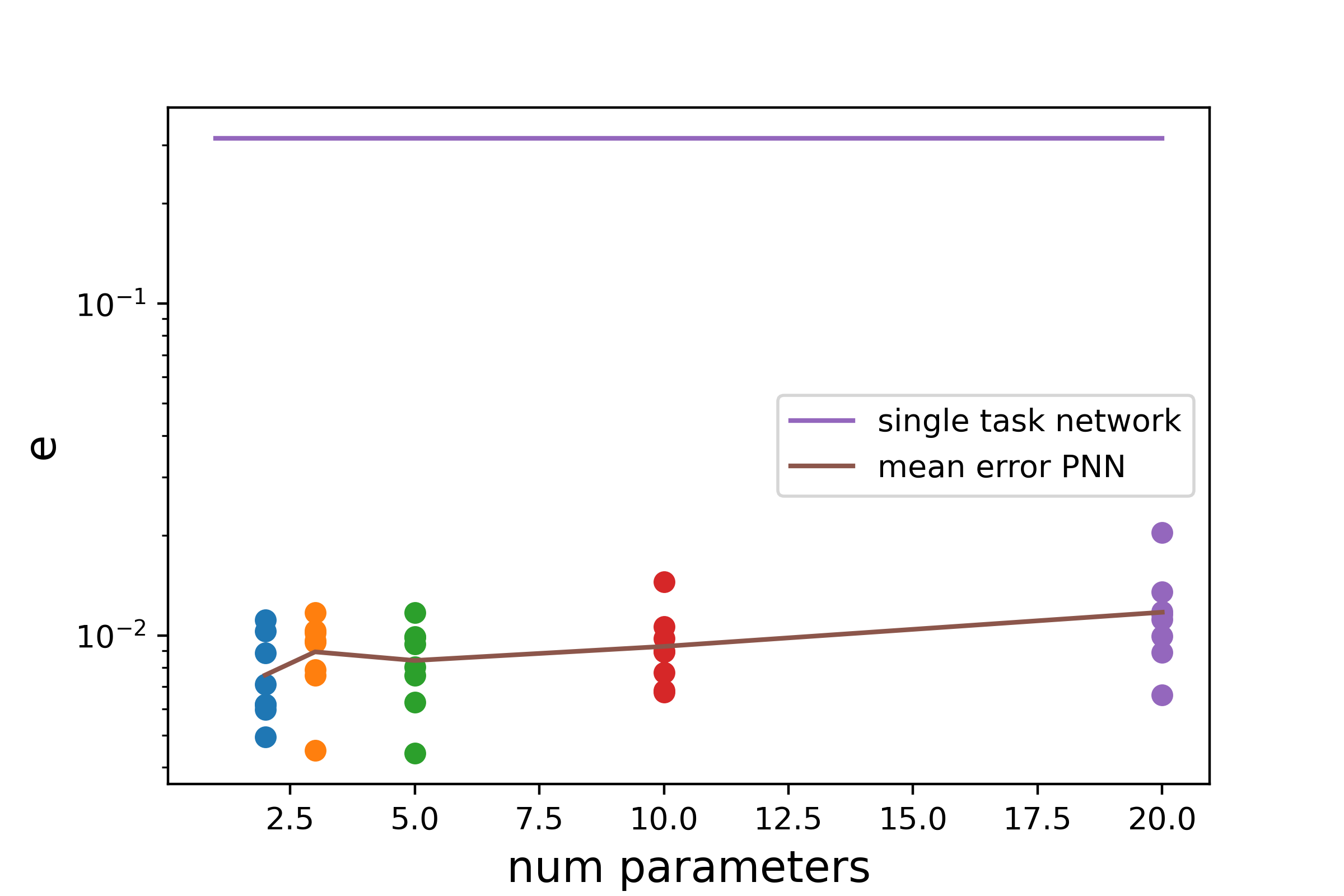}
\includegraphics[width=0.49\textwidth]{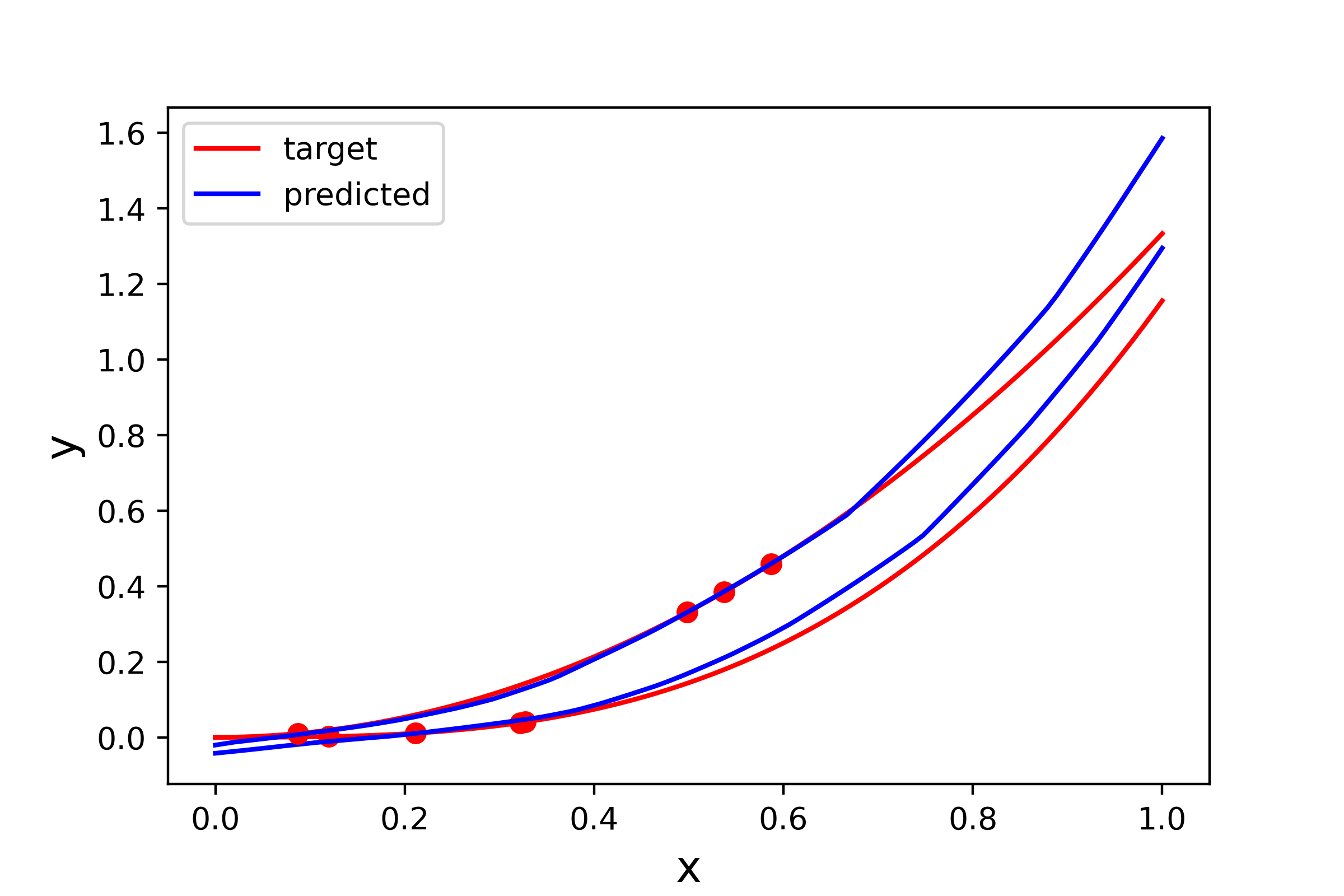}
\caption{Left: Error of PNN estimates for different parameter dimensions. Right: Two sampled tasks, one quadratic and one cubic, together with resulting PNN approximations and true functions.} \label{fig-regime-errors}
\end{figure}
\begin{figure}\centering
\includegraphics[width=0.5\textwidth]{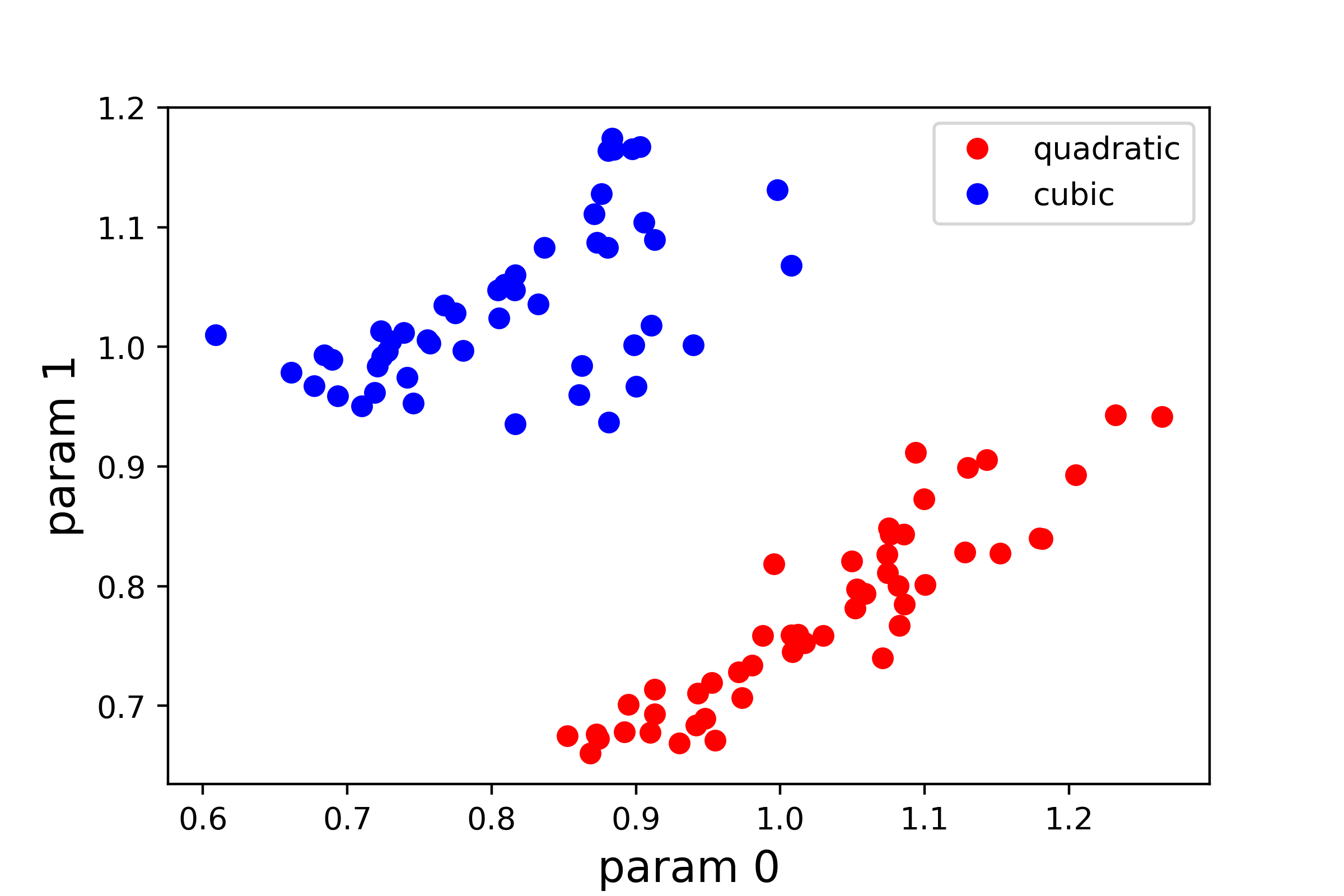}
\caption{Calibrated parameters for all training data points for a PNN with parameter dimension equal to two. } \label{fig-regime-points}
\end{figure}
In this example, the tasks are generated by two different functions and the information on the function used is not encoded in the feature data. The functions are simple quadratic and cubic monomials,
\begin{equation}\label{fig-tasks-regime}
\begin{array}{lcl}
	f_1(x)&=&ax^2 \mbox{ for } x\in[0,1],\\
	f_2(x)&=&ax^3 \mbox{ for } x\in[0,1].
\end{array}
\end{equation}
We construct each task by randomly choosing $f_1$ or $f_2$, sample $a$ uniformly from $[1,2]$ and four $x$ values uniformly from $[0,1]$.
Figure \ref{fig-regime} shows 30 sampled tasks in the left plot, and a sample of 250 functions in the right one. The error relating to different parameter dimensions is given in the left graph of figure \ref{fig-regime-errors}. In this example the error increases for increasing parameter dimension. One may guess that this behavior is due to overfitting, but the training error shows the same pattern. A possible explanation might be, that, since we are not applying any hyperparameter tuning, the optimization did not fully converge to sufficient accuracy. The right graph of figure \ref{fig-regime-errors} shows for a quadratic and a cubic task the predicted values for a PNN (with two parameters), the training data as well as the target functions. For PNNs with parameter dimension two, Figure \ref{fig-regime-points} shows for several tasks the corresponding parameter vector as a scatter plot.
We clearly see that the parameters can be separated into two sets, one representing the quadratic function regime, and the other the cubic function regime. 
\subsection{Bond Spread Curve Calibration}
\begin{figure}
    \includegraphics[width=0.5\textwidth]{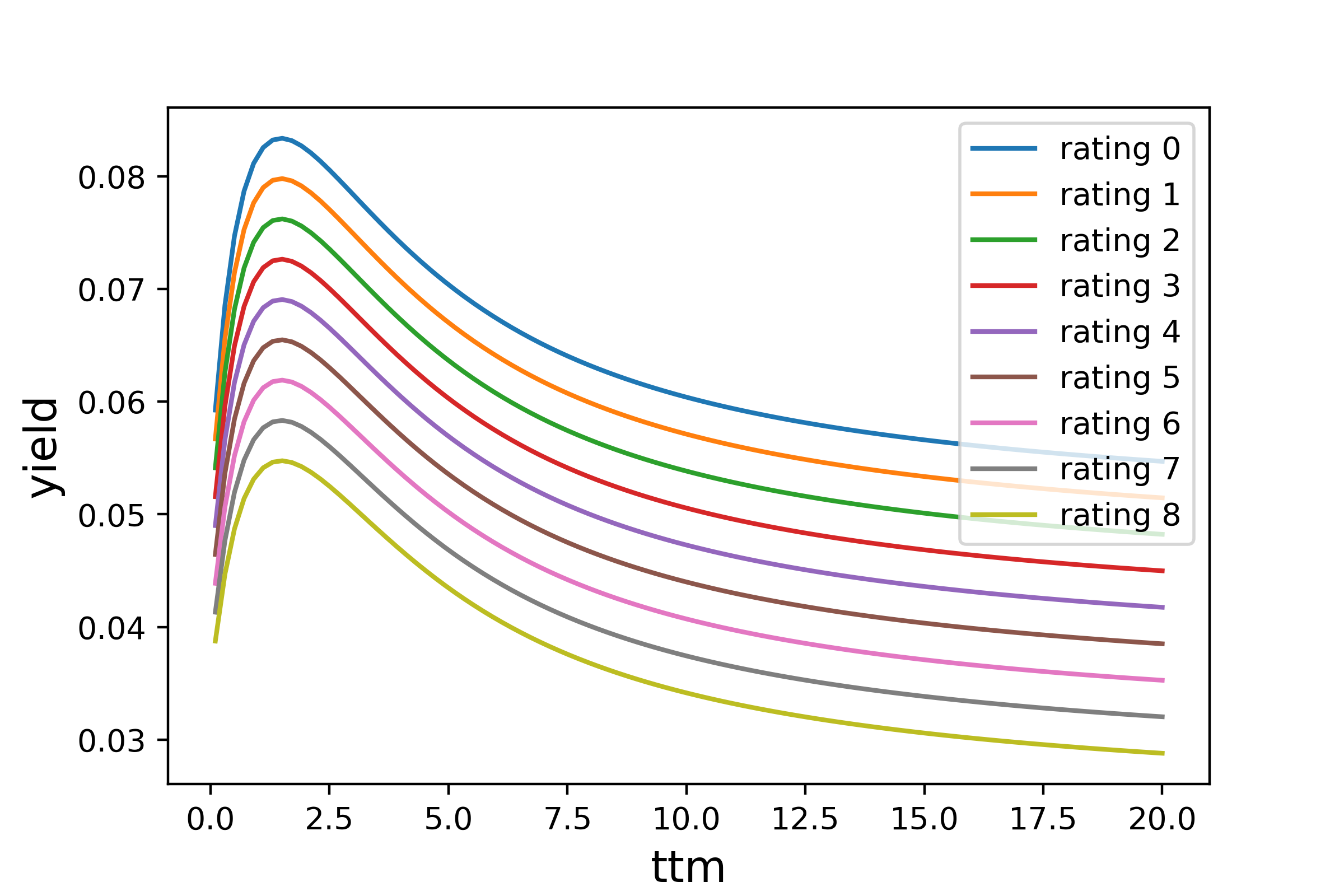}
    \includegraphics[width=0.5\textwidth]{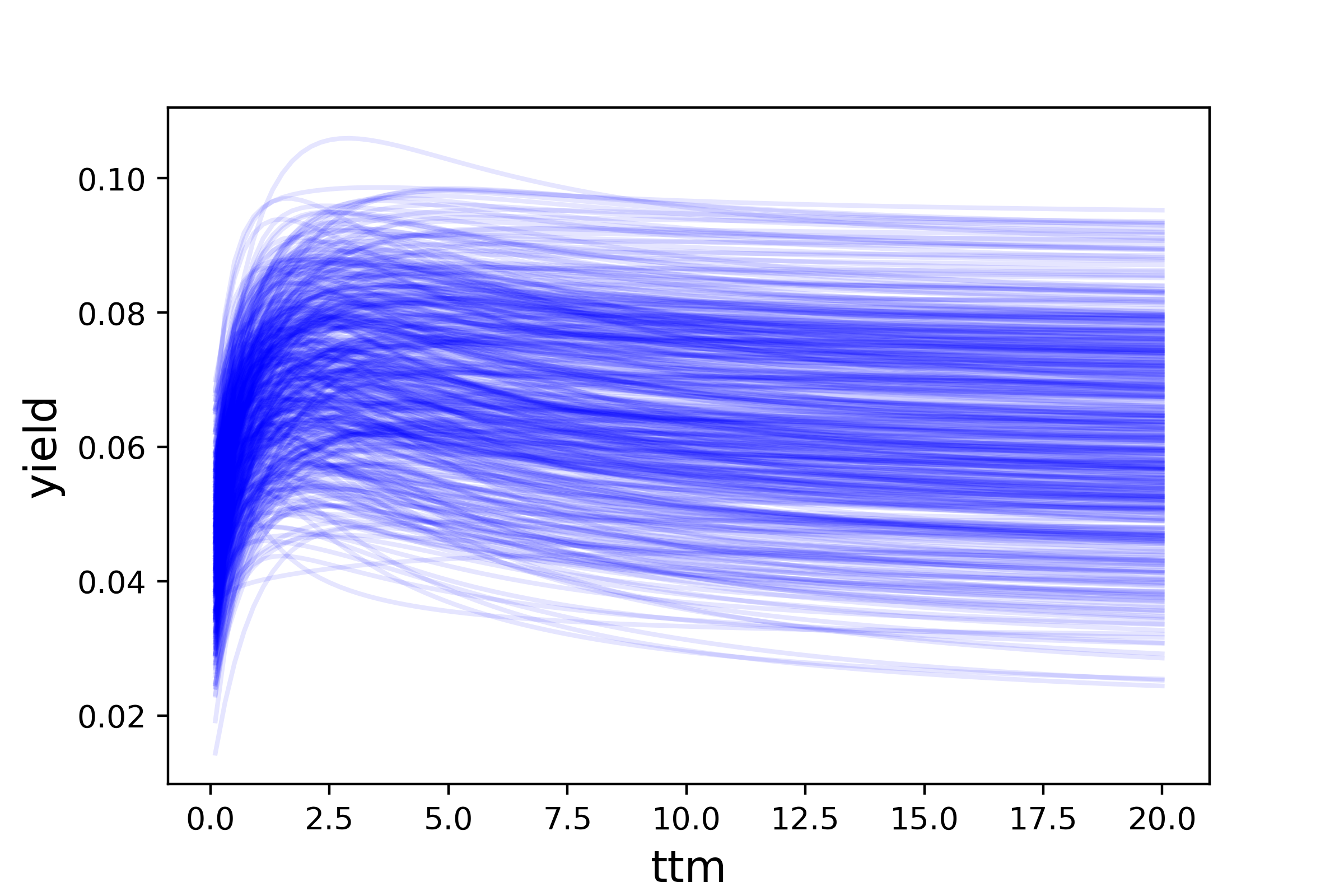}
    \caption{Left: Sampled yield curves for different ratings and all other categories fixed. Right: 500 randomly sampled yield curves with all categories fixed.} \label{fig-yields}
\end{figure}
\begin{figure}
    \includegraphics[width=0.49\textwidth]{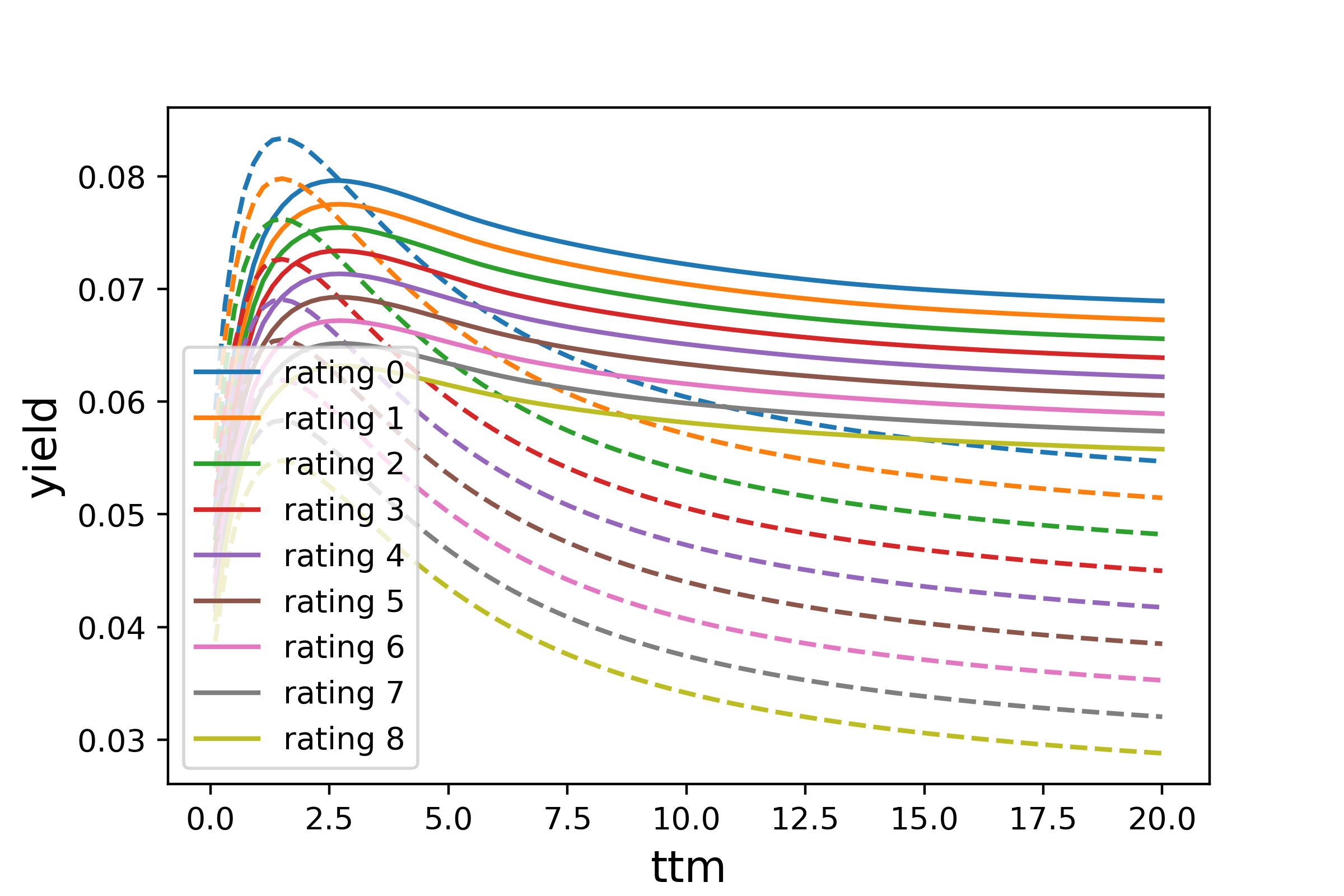}
    \includegraphics[width=0.49\textwidth]{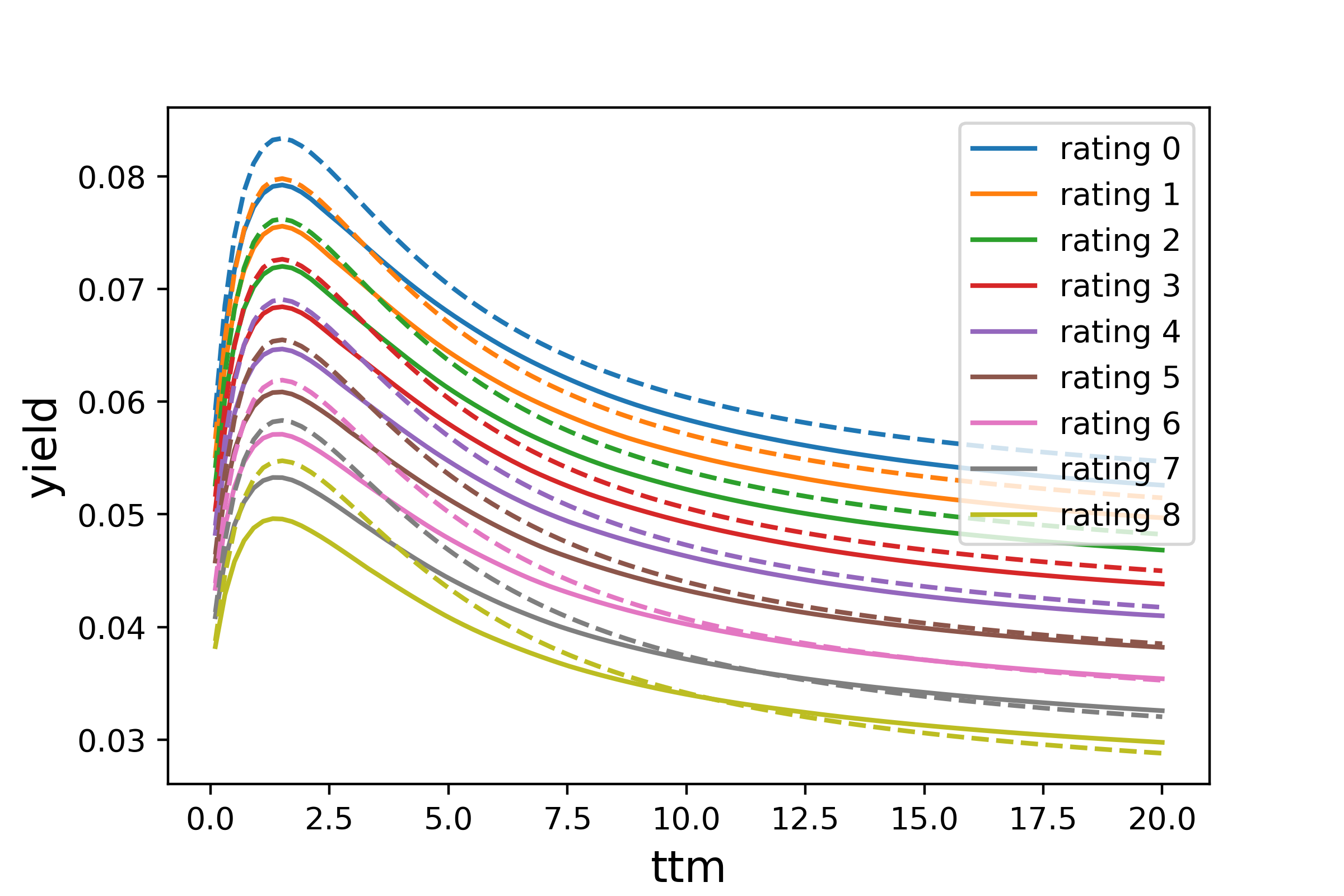}
    \includegraphics[width=0.49\textwidth]{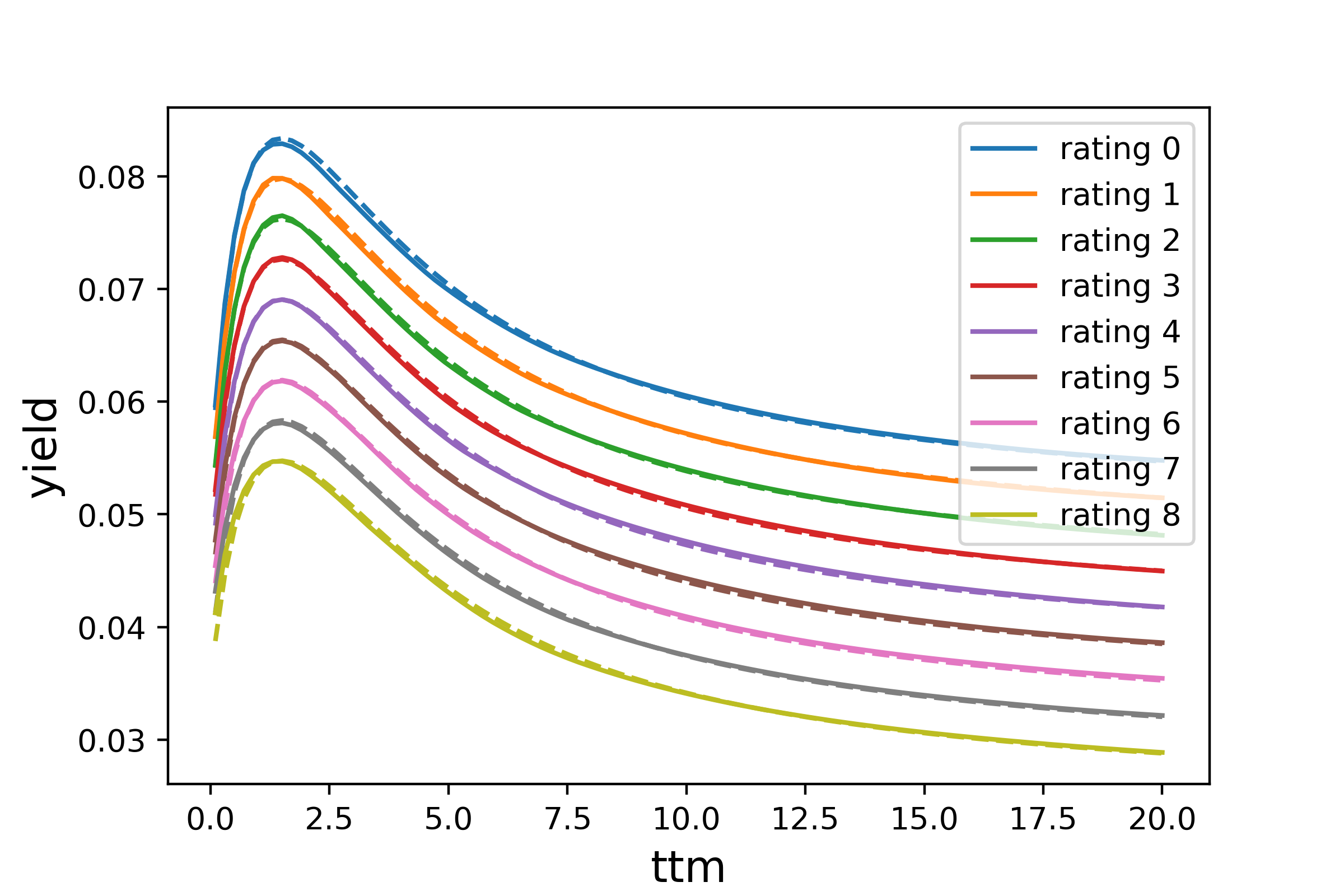}
    \caption{Target (dashed line) and PNN calibrated yield curves for different ratings and a fixed task, initial value (top left), after 1 epoch (top right, only PNN parameters are calibrated) and 15 epochs (bottom middle).} \label{fig-yields-error-evolution}
\end{figure}
\begin{figure} \centering
\includegraphics[width=1\textwidth]{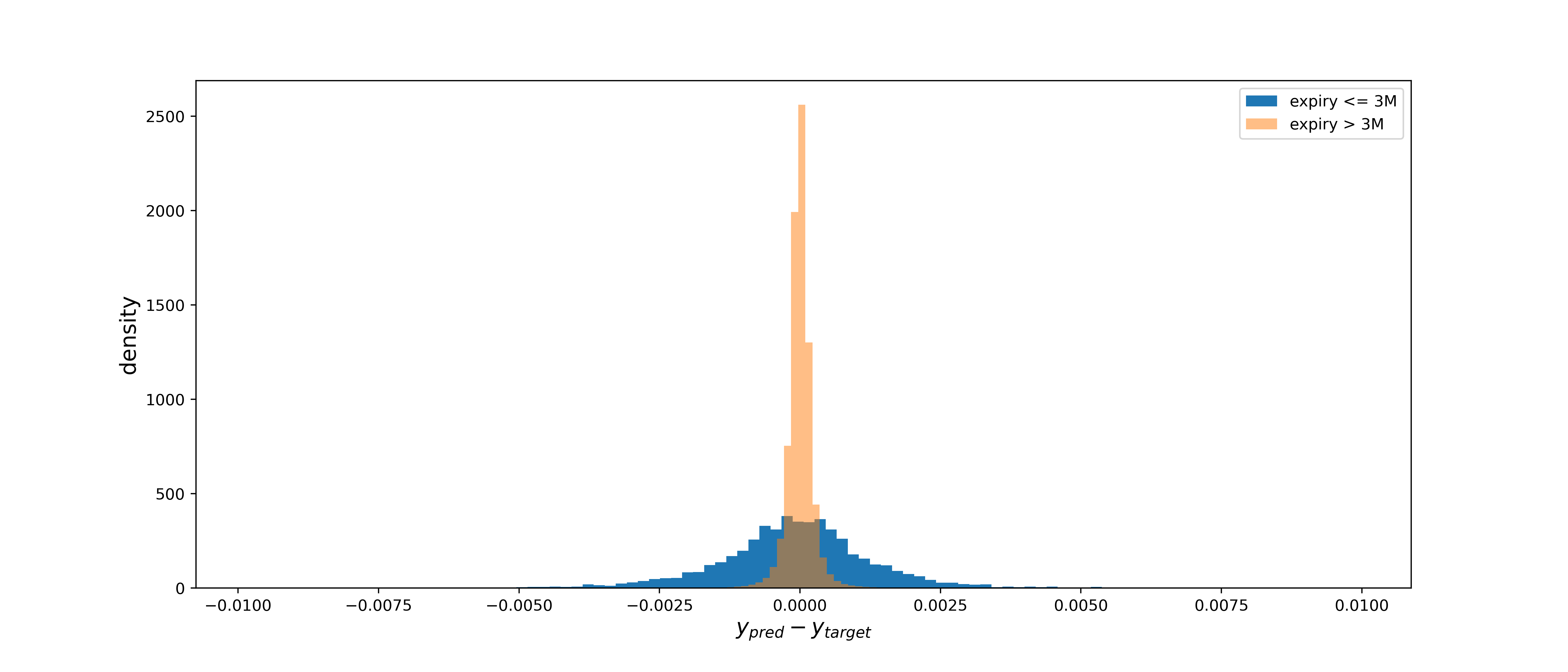}
\caption{Differences between target and predicted values for daily recalibrated PNN (600 days) with 600 different bonds per day.} \label{fig-bond-yield-errors}
\end{figure}
In this section, we analyze the PNN for a  potential application in finance, the calibration of spread curves for bond pricing. Here, we use artificially created data  and not real market data in order to allow for a precise measurement of the performance. The application to real-world data is straightforward and will be presented in future work.

For the pricing of bonds one typically uses curves that express the excess return over a risk-free rate depending on the maturity, the so-called spread curve.  Given such a curve for a bond, the price of the bond can simply be derived by discounting all cash flows of the bond with the values of the curve relating to the dates where the cashflows are paid. 

As the basis for the creation of the data, we use the  
Nelson-Siegel parametrization \cite{Nelson1987} that is given by 
\[
r(T;\beta_0, \beta_1,\beta_2, \tau):= 
        \beta_0 + (\beta_1+\beta_2)\tau (1.0-e^{-T/\tau})/T - \beta_2e^{-T/\tau}.
\]
We assume that there are four different categories influencing the spread: company rating, country, sector, ESG rating. More precisely, we consider 9 rating classes, 5 countries, 11 sectors, and 3 ESG ratings which gives a total of 1485 different classes with different spread curves. Note that other features like liquidity and securitization level that also affect bond prices in practice can also easily be incorporated into our approach. 
In real-world applications, there is usually not enough price data to calibrate all these curves for each category. Therefore, either the categories have to be defined on a coarser level, or some relationships between the curves need to be used. For example, a spread for a bond with a lower rating must be higher than the spread for a bond with the same features but a higher rating. Note that although we may hot have enough quotes on one day to successfully calibrate a network giving the spread for a given bond, 
we usually have a lot of data over time. This allows a PNN to learn a parameterization that reflects relationships between these categories and allows one to calibrate just the embedding parameter to data for a single day. 

The data is created as follows. For each task,  we first sample two sets of Nelson-Siegel parameters for each of the above four categories 
\begin{eqnarray*}
\beta_0^{i,j} &\sim& \mathcal{U}(0,0.15), \\
\beta_1^{i,j} &\sim& \mathcal{U}(0-\beta_0^{i,j},0.1-\beta_0^{i,j}),\\
\beta_2^{i,j}  &\sim& \mathcal{U}(0, 0.2),\\
\tau^{i,j} &\sim&  \mathcal{U}(0.2,2.0),
\end{eqnarray*}
where $i\in\{1,2\}$ and $j$ denotes the category. We then define two curves
\begin{eqnarray}
s_{1,j}(T) &=& r(T;\beta_0^{1,j}, \beta_1^{1,j},\beta_2^{1,j}, \tau^{2,j}),\\
s_{2,j}(T) &=& s_{1,j}(T) + r(T;\beta_0^{2,j}, \beta_1^{2,j},\beta_2^{2,j}, \tau^{2,j}).
\end{eqnarray}
Note, that due to the range of parameters and the construction of the $s_{2,j}$, we always have $s_{1,j}(T)\leq s_{2,j}(T)$. We denote by $n_j$ the number of elements in category $j$, e.g. $n_j=3$ for the category \emph{ESG rating}. We then define the overall curve $s(T;k_1,k_2,k_3,k_4)$, $1\leq k_j\leq n_j$ by
\[
s(T;k_1,k_2,k_3,k_4) =  \sum_{j=1}^4 w_j\left(\frac{n_j-k_j}{n_j-1}s_{1,j}(T) + \frac{k_j-1}{n_j-1}s_{2,j}(T)\right)
\]
where $w_1 = 0.1$, $w_2=0.2$, $w_3=0.5$, $w_4=0.2$. Figure \ref{fig-yields} shows one set of yield curves (left) sampled with different ratings (all other categories are the same) and a set of sampled curves all with the same categories. If we would handle all yield curves separately using the Nelson-Siegel parametrization, due to the high number of different categories we would end up with 5940 parameters. For most interest rate markets there is not enough bond data to calibrate that many parameters. However, due to our construction which imposes a strong structure between the yields of different categories, the problem has in fact just  64 parameters. Here, the PNN may learn this structure using much fewer parameters, making it possible to calibrate the embedded parameters when new data comes in. In our experiment, we use simple  feedforward neural network with three layers and 64 activation functions per layer. Since the credit and ESG ratings exhibit a meaningful ordering, we transform them to ordinal values between zero and one. The maturity is scaled linearly so that 20 years are transformed to 1.0. The country and sector features are one hot encoded, which leads to a total input dimension of 19. The training parameters are given by
\begin{itemize}
    \item 500 different tasks as training data where each task consists of 600 bonds and respective yields as training data.
    \item Simple feedforward neural network with three layers and 64 activation functions per layer.
    \item Adam optimizer with the initial learning rate 0.001 and exponential decay learning rate schedule (decay factor 0.99), 1000 data points per batch, and 20.000 epochs.
\end{itemize}

We test the resulting network by sampling new data (600 days with 600 bond yields per day) and recalibrating just the PNN parameters each day, leaving the base model unchanged. For the recalibration of the parameters (leaving the network weights fixed) we use the average over all parameters derived from the initial training as the start value and apply 15 epochs of Adam optimizer over the 600 data points with 10 points per batch. Note the very fast and robust convergence properties we observed in our experiments. Here, figure \ref{fig-yields-error-evolution} shows the resulting curves versus the targets for a data point (varying ratings) at the beginning of the parameter recalibration, after one, and after 15 optimizer steps. 
A histogram of the differences between target and predicted values on the overall test data is shown in figure \ref{fig-bond-yield-errors}, split into the errors for bonds with maturity less than 3 months and all others. We see that most of the errors are in a range of less then 
10 basis points for the bonds with maturities beyond 3 months while the error of short-dated bonds is in a wider range of around 25 basis points. The reason why short-dated bonds show a higher error can be seen in the right graph of figure \ref{fig-yields} where some sample curves for a fixed bond are depicted which shows that most of the curves are very steep at the beginning while on the other hand, compared to the overall number of bonds only, the training data contains not that many bonds with such a short maturity (approximately 1.2\%). Here, to further improve results for the short-dated instruments we could use oversampling or a different weighting. 

\section{Conclusion and Future Work}
In this paper, we discussed and analyzed a very simple form of MTL where all network weights except the bias in the first layer are shared between different tasks.
We showed by several simple examples that this approach is able to learn a family of models for given data that makes it relatively easy to recalibrate the respective parameters to new data avoiding overfitting. Another important aspect to apply methods in the financial domain is the safety and robustness of the algorithm when retrained on new data. Here, an interesting aspect of the proposed method is  that the validation of the model class can be done in advance by validating the trained model w.r.t. the embedding parameters. So if new data comes in and we need to recalibrate the model by recalibrating just the embedded parameters, we have a strong indication that our model behaves well as long as the embedded parameters stay in the range that was used in the validation.
Moreover, one example showed that the calibrated embedded parameters can be used to identify different regimes indicating that these parameters may also be used to analyze the data that is used for training which might give some further insights into the problem structure.

These results indicate the potential power of this approach within the financial domain where many problems may show a certain macrostructure between different tasks overcoming the problem that we may  not have enough data for a successful calibration of a neural network to a single task. As an example, we investigated the performance of the method within the context of calibrating bond yields to market data on an artificially created toy dataset. Here, the different tasks consisted of bond yields on different days depending on typical static bond data such as credit ratings, ESG ratings, countries, and sectors.
This example showed that the proposed method is able to learn from a set of different tasks a parametrization of the bond yields that can be stable and robustly recalibrated to new data.

However, although the results are quite promising we have to test the method on more benchmark applications as well as on real data. Applications may vary from the estimation of credit default probabilities over the construction of new parameterized volatility surfaces up to portfolio optimization problems and the estimation of conditional probabilities.

\section{Appendix - Proof of Theorem \ref{Generalization-Property}}\label{Appendix-proof}
In this section, we give a sketch of the proof for Theorem \ref{Generalization-Property} which is quite similar to the theory presented in  \cite{Baxter2019a} with just minor modifications to handle our special case. We therefore just present the most relevant theorems and lemmas for the proof of Theorem \ref{Generalization-Property} and refer to  \cite{Baxter2019} for the proof of these statements and for further definitions. Recall from section \ref{sec-task-embedding} that $n$ denotes the number of tasks and $m$ the number of samples per task, let $Z\subset\mathbb{R}^{k_1}\times \mathbb{R}^{k_2}$. 

To deal with the additive structure of (\ref{empirical-loss_mt}) and (\ref{true-loss_mt}) we first define the following.
\begin{definition}
Let $\mathcal{H}_1,$...,$\mathcal{H}_n$ be n sets of functions mapping $Z$ into $[0,M]$, For all $h_i\in\mathcal{H}_i$, define 
\[
    \bigoplus_{i=1}^n h_i (\vec{z}):=\frac{1}{n}\sum_{i=1}^nh_i(z_i)
\]
and define the set of all these functions as $\bigoplus_{i=1}^n \mathcal{H}_i$.
\end{definition}
For such an additive structure we have the following theorem from \cite{Baxter2019a}.
\begin{theorem}\label{theorem-1}
Let $\mathcal{H}\subset \bigoplus_{i=1}^n \mathcal{H}_i$ be a permissible set of functions  $Z^n\mapsto [0,M]$. Let $\mathbf{z} \in Z^{(m,n)}$ be generated by $m>\frac{2M}{\alpha^2\nu}$ independent trials from $Z^n$ according to some product probability measure $\vec{P}=P_1 \times \cdots\times P_n$. For all $\nu>0$, $0<\alpha<1$,
\[
Pr\left\{\mathbf{z}\in Z^{(m,n)}:\exists \vec{h}\in \mathcal{H}: d_\nu\left(\langle \vec{h}_\mathbf{z}\rangle,  \langle \vec{h}_{\vec{p}} \rangle\right)>\alpha\right\}
\leq 4\mathcal{C}(\alpha\nu/8,\mathcal{H})e^{-\frac{\alpha^2\nu n m}{8M}},
\]
where $\langle \vec{h}_\mathbf{z}\rangle:=\frac{1}{m}\sum_{i=1}^mh(\vec{z}_i)$ and $\langle h\rangle_P = \int_{Z^n}h(\vec{z})dP(\vec{z})$
\end{theorem}
From \cite{Baxter2019a} we have the following Lemma.
\begin{lemma}\label{lemma-1}
Let $\mathcal{H}:X\mapsto A$ be of the form $\mathcal{H}=\mathcal{G}\circ \mathcal{F}$ where $X\xmapsto{\mathcal{F}} V \xmapsto{\mathcal{G}} A$. For all $\varepsilon_1,\varepsilon_2>0$, $\varepsilon=\varepsilon_1+\varepsilon_2$,
\begin{equation}\label{Eq-Capacity-Composition}
\mathcal{C}(\varepsilon) \leq \mathcal{C}_{l_{\mathcal{G}}}(\varepsilon_1,\mathcal{F})\mathcal{C}({\varepsilon_2,l_{\mathcal{G}}}).
\end{equation}
\end{lemma}

\begin{lemma}\label{lemma-2}
For the function space $\mathcal{H}=\mathcal{G}\circ\mathcal{F}$ with $X\xmapsto{\mathcal{F}} V \xmapsto{\mathcal{G}} A$, $\mathcal{F}\subset \mathcal{F}_1\times\cdots \times \mathcal{F}_n$ and $\mathcal{G}\subset \mathcal{G}_1\times \cdots \mathcal{G}_n$,
\begin{eqnarray}
\mathcal{C}_{l_{\mathcal{G}}}(\varepsilon,\mathcal{F}) &\leq& \prod_{i=1}^n \mathcal{C}_{l_{\mathcal{G}_i}}(\varepsilon,\mathcal{F}_i), \label{Eq-Capacity-Product}\\
\mathcal{C}(\varepsilon, l_{\mathcal{G}})&\leq& \prod_{i=1}^n\mathcal{C}(\varepsilon,l_{\mathcal{G}_i}).
\end{eqnarray}
\end{lemma}

\begin{lemma}\label{lemma-3}
For $\bar{\mathcal{G}}:=\{(g(x_1), \cdots, g(x_n))\mid g\in \mathcal{G}\}\subset \mathcal{G}^n$ we have
\begin{equation}
\mathcal{C}(\varepsilon,l_{\bar{\mathcal{G}}})\leq \mathcal{C}(\varepsilon, l_{\mathcal{G}})
\end{equation}
\begin{proof}
By definition we have
\[
\mathcal{C}(\varepsilon, l_{\mathcal{G}}):=\sup_{P\in \mathcal{P}_{\mathcal{G}}}\mathcal{N}(\varepsilon,\mathcal{G}, d_P )
\]
and analogously 
\[
\mathcal{C}(\varepsilon, l_{\bar{\mathcal{G}}}):=\sup_{\vec{P}\in \mathcal{P}_{\bar{\mathcal{G}}}}\mathcal{N}(\varepsilon,\bar{\mathcal{G}}, d_{\vec{P}} ).
\]
For $\vec{P}=P_1\times\cdots \times P_n\in \mathcal{P}_{\bar{\mathcal{G}}}$ define $\bar{P}=\frac{1}{n}\sum_{i=1}^n P_i$. 
We show that 
$
\{\bar{g}\mid g \in \mathcal{N}(\varepsilon,\mathcal{G}, d_{\bar{P}} )\}
$ is an $\varepsilon$-cover for $\bar{\mathcal{G}}$.

For $\bar{g}=(g(x_1),\cdots ,g(x_n)) \in \bar{\mathcal{G}}$ fixed we select $\tilde{g} \in \mathcal{N}(\varepsilon,\mathcal{G}, d_{\bar{P}} )$ such that 
$
d_{[\bar{P}, l_{\mathcal{G}}]}(\tilde{g}, g) \leq \varepsilon$ and get
\begin{eqnarray*}
d_{[\vec{P},l_{\bar{\mathcal{G}}}]}(\bar{g}, \bar{\tilde{g}}) &=& \frac{1}{n} \sum_{i=1}^n d_{[P_i, l_{\mathcal{G}}]}(g,\tilde{g})\\
&=& d_{[\bar{P}, l_{\mathcal{G}}]}(g,\tilde{g})\\
&\leq& \varepsilon.
\end{eqnarray*}
\end{proof}

\end{lemma}
\begin{theorem}\label{theo-1}
(C.10) For the structure 
\[ 
X^n\xmapsto{\mathcal{F}^n}V^n\xmapsto{\bar{\mathcal{G}}}A^n
\]
a loss function $l:Y\mapsto [0,M]$, and all $\varepsilon, \varepsilon_1,\varepsilon_2>0$ such that $\varepsilon=\varepsilon _1+\varepsilon_2$,
\[
    \mathcal{C}(\varepsilon, l_{\bar{\mathcal{G}}\circ\mathcal{F}^n}) \leq C(\varepsilon_1,l_{\mathcal{G}}) \mathcal{C}(\varepsilon_2,\mathcal{F})^n
\]
\end{theorem}
\begin{proof}
\begin{eqnarray*}
\mathcal{C}(\varepsilon, l_{\bar{\bar{\mathcal{G}}}\circ \mathcal{F}^n}) &\leq&
    \mathcal{C}(\varepsilon_1,l_{\bar{\mathcal{G}}})\mathcal{C}_{l_{\bar{\mathcal{G}}}} (\varepsilon_2, \mathcal{F}^n) \mbox{ (using (\ref{Eq-Capacity-Composition}}))\\
    &\leq& 
    \mathcal{C}(\varepsilon_1,l_{\bar{\mathcal{G}}})
    \mathcal{C}_{l_{\mathcal{G}}} (\varepsilon_2, \mathcal{F})^n
    \mbox{ (using (\ref{Eq-Capacity-Product}))}
\end{eqnarray*}
\end{proof}
With the results above the proof of theorem \ref{Generalization-Property} that we restate in the following for the sake of completeness is now straightforward.

\begin{theorem} Let $\nu>0$, $0<\alpha<1$, be fixed and $\varepsilon_1, \varepsilon_2 > 0$ such that $\varepsilon_1+\varepsilon_2=\frac{\alpha\nu}{8}$. 
For $0<\delta<1$ and the structure 
\[ 
X^n\xmapsto{\mathcal{F}^n}V^n\xmapsto{\bar{\mathcal{G}}}A^n
\] and $\mathbf{z} \in Z^{(m,n)}$ be generated by 
$m>\frac{8M}{\alpha^2\nu}\left[ \ln(\mathcal{C}(\varepsilon_1,\mathcal{F})) + \frac{1}{n}\ln \frac{4\mathcal{C}(\varepsilon_2,l_{\mathcal{G}})}{\delta} \right]$ independent samples
we have
\begin{equation}
    PR\left\{ z \in Z^{(m,n)}:\exists \bar{g}\circ\vec{f}\in \bar{\mathcal{G}}\circ\mathcal{F}^n: d_{\nu}(\langle l_{\bar{g}\circ \vec{f}}\rangle_{\mathbf{z}}, \langle l_{\bar{g}\circ \vec{f}}\rangle_{\vec{P}}) > \alpha\right\}\leq \delta
\end{equation}
\end{theorem}
\begin{proof}
Using theorem \ref{theorem-1} we have
\[
PR\left\{ \mathbf{z}\in Z^{(m,n)}:\exists \bar{g}\circ\vec{f}\in \bar{\mathcal{G}}\circ\mathcal{F}^n: d_{\nu}(\langle l_{\bar{g}\circ \vec{f}}\rangle_{\mathbf{z}}, \langle l_{\bar{g}\circ \vec{f}}\rangle_{\vec{P}}) \right\} \leq  4\mathcal{C}(\alpha\nu/8,l_{\bar{\bar{\mathcal{G}}}\circ \mathcal{F}^n})e^{-\frac{\alpha^2\nu n m}{8M}}
\]
and using theorem \ref{theo-1} we obtain
\begin{eqnarray}
    4\mathcal{C}(\alpha\nu/8,l_{\bar{\bar{\mathcal{G}}}\circ \mathcal{F}^n}) & \leq & 4 \mathcal{C}_{l_\mathcal{G}}(\varepsilon_1, \mathcal{F})^n  \mathcal{C}(\varepsilon_2, l_{\mathcal{G}}) 
\end{eqnarray}
and simple calculation proves the statement.
\end{proof}


\bibliography{literature}
\bibliographystyle{ieeetr}
\end{document}